\documentclass[10pt]{article}
\usepackage{amsfonts}
\usepackage{amssymb}
\usepackage{amsthm}
\usepackage{amsmath} 
\usepackage{graphicx}
\usepackage{hyperref}
\usepackage{color}
\usepackage[mathscr]{euscript}

\numberwithin{equation}{section}

\newcommand{\ii}{{\rm i}}
\newcommand{\ee}{{\rm e}}
\newcommand{\x}{{\rm x}}

\newcommand{\dd}{{\rm d}}
\newcommand{\beq}{\begin{equation}}
\newcommand{\ene}{\end{equation}}
\newcommand{\ds}{\displaystyle}

\newtheorem{thm}{Theorem}
\newtheorem{lemma}[thm]{Lemma}
\newtheorem{prop}[thm]{Proposition}

\theoremstyle{definition}

\begin{document}

\title{Quantum field theory with dynamical boundary conditions and the Casimir effect \thanks{Research partially supported by project  PAPIIT-DGAPA UNAM  IN103918 and by project SEP-CONACYT CB 2015, 254062. }}
\author{Benito A. Ju\'arez-Aubry\thanks{Fellow Sistema Nacional de Investigadores. Email: benito.juarez@iimas.unam.mx}\, and Ricardo Weder\thanks{Fellow Sistema Nacional de Investigadores. Email:weder@unam.mx. Home page: http://www.iimas.unam.mx/rweder/rweder.html.}\\
Departamento de F\'isica Matem\'atica, \\
Instituto de Investigaciones en Matem\'aticas Aplicadas y en Sistemas, \\ Universidad Nacional Aut\'onoma de M\'exico,\\Apartado Postal 20-126, Ciudad de M\'exico 01000, M\'exico.}
\date{}

\maketitle

\vspace{.5cm}
\centerline{Abstract}

\noindent We   study  a coupled system that describes the interacting dynamics between a {\it bulk} field, confined to a finite region with timelike boundary, and a {\it boundary} observable. In our system the dynamics of the boundary observable prescribes dynamical boundary conditions for the bulk field. We cast our classical system in the form of an abstract  linear Klein-Gordon equation,  in an enlarged Hilbert space for  the bulk field and the boundary observable. This makes it possible to apply to our coupled system the general methods of quantization. In particular, we implement the Fock quantization in full detail. Using this quantization  we study the Casimir effect in our coupled system. Specifically,  we compute  the renormalized local state polarization and the local Casimir energy, which we can define for both the bulk field and  the boundary observable of our system. Numerical examples in which the integrated Casimir energy is positive or negative are presented.


\section{Introduction}
\label{sec:Intro}

In this work we are concerned with the study of a system describing the coupled dynamics between a {\it bulk} field, confined to a region with (timelike) boundary, and a {\it boundary} observable, whereby the dynamics of the boundary observable prescribes the boundary conditions for the bulk field. While the  coupled system can be seen as an interacting one, as we shall see, it can be cast in the form of an abstract, linear Klein-Gordon equation, for which Fock quantization can be implemented in full detail. Once a quantum description of the system is available, a most natural question is to characterize the Casimir effect in this system. Indeed, the purpose of this paper is to study the renormalized local state polarization and the local Casimir energy , which we can define for both the bulk field and the boundary observable of the system.

This class of mixed bulk-boundary systems have received attention both in the mathematics and physics literature.
For time-periodic solutions they correspond to  Sturm-Liouville problems with boundary conditions that depend in the spectral parameter.  In the mathematics literature this type of Sturm-Liouville problems has been extensively studied. See for example \cite{walter}, \cite{fulton:1977}, \cite{shk96}, \cite{ben99},\cite{tretter}, \cite{mar}, \cite{menn}, and the references quoted there.

In the physics literature, such bulk-boundary systems have been studied recently in the context of quantum field theory in \cite{G.:2015yxa, Barbero:2017kvi, Dappiaggi:2018pju, Zahn:2015due} owing to different motivations. In \cite{G.:2015yxa, Barbero:2017kvi}, a central motivation was to underpin precisely the notion of boundary degrees of freedom, especially in connection with the study of black hole entropy calculations in the context of isolated or dynamical horizons \cite{Ashtekar:2004cn, Ashtekar:1997yu}, and also to understand whether boundary degrees of freedom could serve as particle detectors (in the spirit of \cite{Unruh:1976db}) for bulk quantum fields. On the other hand, the works \cite{ Dappiaggi:2018pju,Zahn:2015due} have been largely inspired by the holographic programme of AdS/CFT and its generalisations (see e.g. \cite{Maldacena:1997re, Witten:1998qj}), and the similarities between these bulk-boundary systems and so-called holographic renormalization\cite{Skenderis:2002wp}. 
It is further argued in \cite{G.:2015yxa, Barbero:2017kvi, Dappiaggi:2018pju} that the techniques applied to the study of bulk-boundary systems should also prove relevant for high-energy physics in the context of Yang-Mills in lower dimensions and Chern-Simons theory \cite{Witten:1988hf} and Maxwell-Chern-Simons theory \cite{Schonfeld:1980kb, Deser:1982vy,Deser:1981wh}. Such techniques are presumably also relevant for condensed matter theory. See e.g. \cite{Martin-Ruiz:2015skg} in the context of the modelling of topological insulators using effective field theory and \cite{Franca:2019twk} for the study of effects occurring therein.

In the physical literature discussed above, the focus has been mainly on the construction of states and linear observables in quantum theory for such bulk-boundary theories. This work differs in two senses. First, the dynamical boundary conditions that we impose are more general. In particular, the coefficient $\beta_2'$ that appears in \eqref{Dyn} is set to zero in the previous literature, whereas here we allow it to be different from zero. When $\beta_2'$ is nonzero there is a new physical effect, namely an interaction between the bulk and the boundary that depends on the second time derivative of the bulk field, that is to say, an interaction that is frequency dependent.  Second and most importantly, the main focus in this work is to study the Casimir effect, which is one of prime interest from a theoretical and experimental viewpoint. 

The Casimir effect posits that a system confined to some finite physical region must have a non-trivial vacuum energy density, which will depend on the boundary conditions of the system, even if the spacetime geometry is flat. For this reason, this effect has also been of great interest in the quantum field theory literature. The non-trivial vacuum energy is referred to as the Casimir energy. See e.g. \cite{Bordag}  and \cite{Dodonov:2020eto}  for a literature overview on the subject.
To the best of our knowledge, the first clear-cut calculation, which clarifies the origin of the Casimir effect in full quantum field theory, was performed by Kay in \cite{Kay:1978zr} in the case of a scalar field with periodic boundary conditions -- i.e., in a cylindrical universe. A thorough discussion of the effect in the same spirit is presented in \cite[Chap. 5]{Fulling:1989nb}. In those references, a point-splitting regularization and Minkowski-vacuum subtraction are used to define the local Casimir energy. Our strategy to compute the local Casimir energy is in this spirit too, only differing in the fact that we perform a Hadamard subtraction instead (see Sec. \ref{subsec:Hadamard} below), which is equivalent to removing the Minkowski vacuum contribution in our case of interest. Note that it is now understood that the Hadamard subtraction is better suited for generalizing renormalization prescriptions to curved spacetimes. See e.g. \cite{Wald:1995yp}. 

To the best of our knowledge, the only instance in which the Casimir effect has been studied in the context of dynamical boundary conditions, i.e. for bulk-boundary systems, appears in  \cite{Fosco:2013wpa}. That work is motivated by the fact that superconducting circuit experiments, which are relevant to the experimental measurement of the Casimir energy, are modelled most appropriately by using dynamical boundary conditions \cite{Wilson:2011}.

Our works differs from \cite{Fosco:2013wpa} in that we are interested in the local Casimir energy, both at zero temperature and at positive temperature, while they focus on the integrated Casimir energy at zero temperature, and in that the boundary condition that we impose (i.e. the equation for the boundary observable) is more general as we allow for $\beta'_2 \neq 0.$ The method used for obtaining the Casimir energy is also different. In \cite{Fosco:2013wpa} the regularization procedure is implemented by confining the system to a large box, and obtaining the difference between the total energy of the system of interest (with dynamical boundary conditions) and the would-be total energy for the system confined in the large box (with non-dynamical boundary conditions), with the aid of complex-analytic techniques. The integrated Casimir energy is then obtained by taking the infinite-size limit of the large-box reference system. The advantage of this method is that it allows them to study the static and dynamical Casimir effects in an efficient way, but this comes at the price of having no information about the local Casimir energy, which we obtain in this paper. 

 At temperature zero the main results for the renormalized  local state polarization are given at   \eqref{DirichletVacuumBulk},  \eqref{DirichletVacuumBound},   \eqref{RobinVacuumBulk}, and \eqref{RobinVacuumBoundary}, and for the local Casimir energy at  \eqref{HDiriBulk}, and  \eqref{DiriHBound}  \eqref{RobinHBulk}, and \eqref{RobinHBoundary}.   At positive temperature the main results for the renormalized  local state polarization are given at  \eqref{TBulkPola} and \eqref{TBoundPola}, and for the local Casimir energy at \eqref{HTbulk} and \eqref{HTboundary}.



The paper is organized as follows. In Section ~\ref{sec:Class} we consider our classical system, and we formulate it as an abstract Klein-Gordon equation. In Section~\ref{sec:QFT} we quantize our classical system, we consider the Hadamard property and we introduce the renormalized local state polarization and the local Casimir energy.  
In Section~\ref{Casimir} we obtain our results on the renormalized local state polarization and on the local Casimir energy at zero temperature. In Section~\ref{CasimirTemp}   we obtain our results on the renormalized local state polarization and on the local Casimir energy at positive temperature. In Section~\ref{sec:Numerics} we present numerical examples in which the integrated Casimir energy can be positive or negative. Sec. \ref{Sec:Conclusions} contains our final remarks. In Appendix~\ref{app:UsefulFormulae} we  state some formulae that we use. In Appendix~\ref{App:Estimates} we obtain estimates on the eigenvalues of our classical problem. Finally, in Appendix~\ref{app:Rob} we give  details of the calculation of the renormalized local state polarization and of the local Casimir energy. 
\section{The classical problem}
\label{sec:Class}

\
We consider a scalar field obeying the following dynamical equation  $\phi: \mathbb R \times [0,\ell] \to \mathbb{R}$, where 
$\ell >0.$
\begin{align}
\label{Dyn}
\left\{
                \begin{array}{l}
                  \left[\partial^2_t- \partial^2_z + m^2 + V(z) \right] \phi(t, z) = 0, \, t \in \mathbb{R},  z \in (0,\ell),  \\
                  \cos\alpha \,\phi(t, 0) + \sin\alpha \,\partial_z \phi(t,0) = 0, \, \alpha \in [0, \pi), \\
                  \left[\beta^\prime_1 \partial^2_t -\beta_1 \right] \phi(t, \ell) = - \beta_2 \partial_z \phi(t,\ell) + \beta'_2 \partial_z \partial^2_t \phi(t,\ell).
                \end{array}\right.
\end{align}
Here, $ m^2 >0$ is a mass parameter, and $\beta_1, \beta_2, \beta_1^\prime, \beta_2^\prime$ are real parameters. The parameter $\beta^\prime_1$ can be seen as the square of an inverse velocity, $\beta_1$ as a mass or as a constant potential and,  $\beta_2,$  $\beta_2^\prime$ can be viewed as coupling parameters to external sources for the boundary dynamical observable $\phi_\partial(t):= \phi(t,\ell).$ Furthermore, the potential $V(z)$ is a real valued continuous function defined for  $ z \in [0,\ell].$ 
The system  \eqref{Dyn} subject to initial data on the surface defined by $t = 0$ can  be viewed as a dynamical   system that describes the interaction between a bulk field $\phi(t, z),  0 < z < \ell,$    and a boundary observable $\phi_\partial(t):=\phi(t,l) .$   Let us consider a solution to \eqref{Dyn} of the form,
\beq\label{w.1}
\phi(t, z)= e^{-i\omega t}\, \varphi(z).
\ene
Inserting \eqref{w.1} into \eqref{Dyn} we obtain,
 \begin{align}
\label{w.2}
\left\{
                \begin{array}{l}
                  \left[- \partial^2_z + m^2 + V(z) \right] \varphi( z) =  \omega^2\, \varphi(z),  z \in (0,\ell),  \\
                  \cos\alpha \,\varphi(0) + \sin\alpha \,\partial_z \varphi(0) = 0, \, \alpha \in [0, \pi), \\
                  -\left[ \beta_1  \varphi( \ell)  - \beta_2 \partial_z \varphi(\ell)\right]    =  \omega^2 \left[\beta_1^\prime \varphi(\ell)  -\beta'_2 \partial_z  \varphi(\ell)\right].
                \end{array}\right.
\end{align}
The  system \eqref{w.2} is a two-point boundary-value problem where the spectral parameter $\omega^2$ appears in the boundary condition at $ z=½\ell.$ Systems of this type appear, after separation of variables, in the one dimensional wave and heat equations, in a  variety of physical problems. In particular, cases wit $\beta^\prime_2\neq 0$ appear in the cooling of a thin solid bar placed at time zero in contact with a finite amount of liquid. See \cite{fulton:1977} for a discussion of these applications.

It is a general fact that  boundary-value problems where the spectral parameter  appears in the boundary condition can be understood as boundary-value problems that describe part of a physical system, and that when the space of states is enlarged to include the missing part of the system, then, the problem can be formulated in an operator theoretical way, with a boundary condition that does not depend on the spectral parameter. This is precisely the situation here, when one enlarges the space of states to include the boundary observable, as we proceed to do now.
 Following \cite{walter},  and\cite{fulton:1977}, we write the boundary value problem \eqref{w.2} in an operator theoretical way with  an associated selfadjoint operator    in a Hilbert space that describes the bulk field and the boundary observable. For this purpose we consider the following extended Hilbert space, $\mathcal H,$ that contains the  boundary observable. Namely,
\beq
\mathcal H:= L^2((0,\ell))\oplus \mathbb C,
\ene
where $\mathbb C $ denotes the complex numbers. $\mathcal H$ is   equipped with the following scalar product.
\begin{equation}\label{Inner}
(\varphi,\chi)_\mathcal{H} := \int_0^\ell \! \dd z \, \varphi_1(z) \overline{\chi_1(z)} +  \rho^{-1} \varphi_2 \overline{\chi_2},
\end{equation}
for $\varphi = (\varphi_1, \varphi_2)^\top$, $\chi = (\chi_1, \chi_2)^\top \in \mathcal{H}$  and where,  
 \begin{align}\label{w.3}
\rho := \det \begin{pmatrix}
           \beta_1' & \beta_1  \\
           \beta_2' & \beta_2 \\
         \end{pmatrix} =  \beta_1^\prime\, \beta_2 - \beta_1\, \beta_2'.  
\end{align}
We always assume that $ \rho >0.$

As in \cite{fulton:1977} let us define the  following   linear operator $A$ in $\mathcal{H}$ by
\begin{align}
\label{opA}
\varphi \in D(A)  \mapsto A\varphi := \begin{pmatrix}
           \left[ - \partial_z^2 + m^2 + V(z) \right] \varphi_1(z)   \\
           -\left[\beta_1 \varphi_1(\ell) - \beta_2 \partial_z  \varphi_1(\ell) \right] \\
         \end{pmatrix},
\end{align}
where the domain of $A$ is defined as follows,

\begin{align}
\label{DomA}
D(A) & = \left\{ \varphi =  \begin{pmatrix} 
           \varphi_1   \\
           \varphi_2  \\
         \end{pmatrix} 
\in \mathcal{H}:  \varphi_1, \partial_z \varphi_1  \text{ are absolutely} \,
 \text{continuous in } [0, \ell],  \partial_z^2 \varphi_1 \in L_2([0,\ell]), \right. \nonumber \\
&\left. \cos\alpha \varphi_1(0) + \sin\alpha \partial_z \varphi_1(0) = 0, 	 \varphi_2 =  \beta_1^\prime\varphi_1(\ell) - \beta_2^\prime\partial_z \varphi_1(\ell) \right\}.
\end{align}
It is proven in \cite{fulton:1977} that if $V$ is continuous in $[0,l]$ and $\rho >0,$  the operator $A$ is densely defined and selfadjoint. Furthermore, \cite{fulton:1977} proves that the spectrum of $A$ is discrete and that it consists of multiplicity one eigenvalues that accumulate at $+\infty$.  The two-point boundary-value problem \eqref{w.2} is equivalent to the following eigenvalue problem for the selfadjoint  operator $A$ in the Hilbert space $\mathcal H$
\beq\label{w.4}
A \varphi= \omega^2\, \varphi, \qquad   \varphi \in D(A).
\ene
Then, the dynamical  system \eqref{Dyn} can be formulated as the following abstract Klein-Gordon equation,

\beq\label{w.5}
\partial^2_t \varphi(t,z) + A \varphi(t,z)=0, \qquad \varphi(t,z) \in D(A).
\ene
To apply to the abstract Klein-Gordon equation \eqref{w.5} the general methods of quantization it is crucial that the operator $A$ be positive. In the next proposition we give sufficient conditions for the positivity  of $A.$  

\begin{prop}\label{pos}
\end{prop} 
Suppose that the determinant $\rho$  defined in \eqref{w.3} is positive,  that $m^2+ V(x) \geq 0, x \in [0,\ell],$  that $ \beta^\prime_2 \neq 0,$  and that either $\alpha =0,$ or $\frac{\pi}{2}  \leq \alpha < \pi.$ Further, assume.
\begin{enumerate}
\item
If $\beta_2^\prime >0,$ then,  $ \beta_1 \geq 0,  \beta_2 < 0,$ and $\beta^\prime_1 < 0.$
\item
If $\beta^\prime_2 < 0,$ then, $ \beta_1 \leq 0,  \beta_2 > 0,$ and $\beta^\prime_1 > 0.$

\end{enumerate}
Under these conditions the eigenvalues of $A$ are all positive, and denoting by $\omega_1^2$ the smallest eigenvalue, $A \geq \omega_1^2 >0.$

\noindent {\it Proof:}
Suppose that $A$ has a nonpositive eigenvalue $ - \lambda,$ with $ \lambda \geq 0,$ and with eigenvector $\varphi= (\varphi_1,\varphi_2)^\top\in D(A),$
\beq\label{w.6}
A\,\varphi = - \lambda\, \varphi.
\ene
We denote
$$
\widehat{\cot \alpha}:= 0,\,  \text{\rm if }\,  \alpha =0,  \text{\rm and}   \, \widehat{\cot \alpha} = \cot \alpha, \, \text{\rm if} \, \frac{\pi}{2} \leq \alpha  < \pi.
$$
Using the first line in the right-hand side of  \eqref{opA}, \eqref{w.6},  and integrating by parts we obtain
\begin{align}\label{w.7}
-\lambda (\varphi_1,\varphi_1)_{L^2((0,\ell))}= ((A \varphi)_1, \varphi_1)_{L^2((0,\ell))} = (\partial_z\varphi_1, \partial_z \varphi_1)_{L^2((0,\ell))}+\\
\left( \left(m^2+V\right) \varphi_1,\varphi_1 \right)_{L^2((0,\ell))} -\widehat{\cot \alpha} |\varphi_1(0)|^2 -
\partial_z\varphi_1(\ell) \overline{\varphi_1(\ell)},
\end{align}
where we also use that if $\alpha =0, \varphi_1(0)=0,$  and that  $ \partial_z\varphi_1(0)= - \widehat{\cot \alpha} \varphi_1(0),$ if  $\frac{\pi}{2} \leq \alpha  < \pi.$ 
By  the  second  line in the right-hand side of  \eqref{opA}, \eqref{w.6}, since  by \eqref{DomA},  $\varphi_2= \beta_1^\prime \varphi_1(\ell)- \beta_2^\prime \partial_z \varphi_1(\ell),$
and reordering the terms  we get,
\beq\label{w.8}
(\beta_2-\lambda \beta^\prime_2) \partial_z\varphi_1(\ell) = (\beta_1- \lambda \beta_1^\prime)\varphi_1(\ell).
\ene 
If  $(\beta_2-\lambda \beta^\prime_2) = (\beta_1- \lambda \beta_1^\prime)=0,$ it follows from \eqref{w.3} that $ \rho=0.$  Since we assume that $\rho\neq 0,$ 
we have that either, $(\beta_2-\lambda \beta^\prime_2) \neq 0$ or $ (\beta_1- \lambda \beta_1^\prime)\neq 0.$ If   $(\beta_2-\lambda \beta^\prime_2) = 0,$ and $ (\beta_1- \lambda \beta_1^\prime)\neq 0$ it follows from \eqref{w.8} that $\varphi_1(\ell)=0.$ Further, if $(\beta_1-\lambda \beta^\prime_1) = 0,$ and $ (\beta_2- \lambda \beta_2^\prime)\neq 0$ it follows from \eqref{w.8} that $\partial_z\varphi_1(\ell)=0.$ Then, in these two cases we have,
\beq\label{bbb.1}
 -\partial_z\varphi_1(\ell) \overline{\varphi_1(\ell)}= 0.
 \ene
Moreover, if  $(\beta_2-\lambda \beta^\prime_2) \neq 0$ and $ (\beta_1- \lambda \beta_1^\prime)\neq 0,$ it follows from \eqref{w.8},
\beq\label{w.9}
\partial_z\varphi_1(\ell)= - \gamma \varphi_1(\ell),
\ene
where
\beq\label{w.10}
\gamma=  \frac{\beta_1- \lambda \beta_1^\prime}{\lambda \beta^\prime_2-\beta_2}.
\ene
Moreover, under assumptions  $(1)$ and $(2)$  above we see that $ \gamma >0.$
Introducing \eqref{bbb.1}  and  \eqref{w.9} into \eqref{w.7} we obtain,

\begin{align}\label{w.11}
-\lambda (\varphi_1,\varphi_1)_{L^2((0,\ell))}=  (\partial_z\varphi_1, \partial_z \varphi_1)_{L^2((0,\ell))}+\\
\left( \left(m^2+V\right) \varphi_1,\varphi_1 \right)_{L^2((0,\ell))} -\widehat{\cot \alpha} |\varphi_1(0)|^2 +\nonumber
\tilde{ \gamma} |\varphi_1(\ell)|^2,
\end{align}
where, either $\tilde{\gamma}=0$, or $\tilde{\gamma}= \gamma >0.$
As the left-hand side of \eqref{w.11} is  nonpositive and the right-hand side is nonnegative it follows that $\varphi_1=0.$ Furthermore, as by  \eqref{DomA}
$ \varphi_2 =  \beta_1^\prime\varphi_1(\ell) - \beta_2^\prime\partial_z \varphi_1(\ell),$ we get $\varphi_2=0,$ and in consequence $\varphi =0.$ This completes the proof that $A$ has no nonpositive eigenvalues. Note that as the eigenvalues of $A$ accumulate at infinity, there can not be a sequence of eigenvalues that converges to zero.

\qed

Assuming that  $ A >0,$  the solution to the abstract Klein-Gordon  equation \eqref{w.5} subject to initial conditions,
\begin{align}\label{w.12}
\left\{
                \begin{array}{l}
                  \phi(0,z) = f(z), \\
                  \partial_t \phi(0,z)  = p(z),
                \end{array}
\right.
\end{align}
is then given by
\begin{equation}\label{GenSolClass}
\phi(t) = \cos(t A^{1/2}) f + \frac{\sin (t A^{1/2})}{A^{1/2}} p.
\end{equation}
For a strong solution we require that $ f \in D(A)$ and $ p \in D(A^{1/2}).$ Note, however, that the right hand side of \eqref{GenSolClass} is well defined for $f,p \in \mathcal H,$ and  that in  this case it defines a weak solution to \eqref{w.5}.

More explicitly, let $\{\Psi_{n}\}_{n\in \mathbb{N}}$ be a  complete set of orthonormal eigenfunctions of  the selfadjoint operator  $A.$
Then, the solution given by \eqref{GenSolClass} takes the explicit form
\begin{align}\label{w.13}
\phi(t,z) = \sum_{n = 1}^\infty \, \Psi_{n}(z) \left[( f, \Psi_{n})_\mathcal{H} \cos(\omega_{n} t) + ( p, \Psi_{n})_\mathcal{H} \frac{\sin(\omega_{n} t)}{\omega_{n}} \right].
\end{align}
The series in \eqref{w.13} converges in the norm of $\mathcal H.$ Note that each of the eigenmodes that appear in \eqref{w.13}, namely,
$$
 \Psi_{n}(z)\, \cos(\omega_{n} t), \, \text{\rm and}\,   \Psi_{n}(z)\,  \frac{\sin(\omega_{n} t)}{\omega_{n}}, \qquad n=1,\dots,
$$
gives a solution to the dynamical system \eqref{Dyn}.

\section{Quantum field theory}
\label{sec:QFT}

As we have seen in Sec. \ref{sec:Class}, problem \eqref{Dyn} can be cast in the form of an abstract  linear Klein-Gordon equation.  In this section we always suppose that the   operator $A$ is positive, that is to say, that all its eigenvalues are positive. In particular, this is true if the assumptions of Proposition~\ref{pos}  hold. It follows that the standard  methods of quantization are available to yield the quantum counterpart of our classical problem. For details, the interested reader can look at refs. \cite{Arai:2018wks, Derezinski:2013dra, Fulling:1989nb, Wald:1995yp}. In \cite{BarberoG.:2019rfb}, special attention is paid to the r\^ole of boundary conditions in canonical quantization.
For our purposes it is convenient to use the canonical quantization method.
The one-particle Hilbert space of the quantum theory is the Hilbert space $\ell^2(\mathcal N)$, that consists of all complex-valued, square-summable sequences, 
$\{ \alpha_n \}_{n=1}^\infty$, with the scalar product,
$$
\left(  \{ \alpha_n \}_{n=1}^\infty, \{ \beta_n \}_{n=1}^\infty  \right)_{l^2(\mathcal N)}:= \sum_{n=1}^\infty\, \alpha_n\, \overline{\beta_n}.
$$
The   many-particle Hilbert space of the quantum theory is the   bosonic Fock space  \cite{Arai:2018wks} 
$$
\mathscr{H} = \mathbb{C} \oplus_{n = 1}^\infty   ( \otimes_s^n \ell^2(\mathcal N)),
$$ 
where  $\otimes_s^n \ell^2(\mathcal N)$denotes  the symmetric tensor product of $n$ copies of   $\ell^2(\mathcal N), n=1,\dots$ \cite{Arai:2018wks}.   

In describing the system at zero temperature, there is a distinguished ground state -- the vacuum state, 
$$\Omega_\ell = (1,0,0,\cdots) \in \mathscr{H}.
$$
The quantum fields are operators on $\mathscr{H}$ given by 
\begin{align}
\hat{\Phi}(t,z) = \sum_{n = 1}^\infty \frac{1}{(2 \omega_n)^{1/2}} \left(\ee^{-\ii \omega_n t} \Psi_n (z)\, {a}_n + \ee^{\ii \omega_n t} \overline{\Psi_n (z)}\, {a}_n^\dagger \right),
\end{align}
where $ a_n$ and $ a_n^\dagger$ are annihilation and creation operators on Fock space, satisfying the canonical commutation relations, 
$[a_n, a_{l}^\dagger] = \delta_{n,l }, [a_n, a_{l}] = 0,[a_n^\dagger, a_{l}^\dagger] =0, n,l =1,\dots,$   $\{\Psi_{n}\}_{n\in \mathbb{N}}$ is a  complete set of orthonormal eigenfunctions of  the selfadjoint operator  $A$ introduced in Sec. \ref{sec:Class}, and $\omega_n^2$ is the eigenvalue that corresponds to $\Psi_n, n=1,\dots.$ We write the eigenfunctions as follows \cite{fulton:1977},

\begin{align}
\Psi_n(z) = \begin{pmatrix}
           \psi_n(z)  \\
           \beta'_1 \psi_n(\ell) - \beta_2'   \partial_z\psi_n(\ell)
         \end{pmatrix}.
\label{Eigenfunctions}
\end{align}
The detailed form of the eigenfunctions $\Psi_n, n=1,\dots,$ depends on the boundary condition imposed at $z = 0$. See below in Sec. \ref{subsec:Dirichlet} and \ref{subsec:Robin}. The creation operator $a_n^\dagger$ creates a particle in the state $\Psi_n, n=1,\cdots,$ and the annihilation operator $a_n$ annihilates a particle in the state $\Psi_n, n=1,\dots.$

The  two-point  Wightman function in the vacuum $\Omega_\ell$ is given by
\begin{align}
\langle \Omega_\ell | \hat \Phi(t,z) \, \hat \Phi(t',z') \Omega_\ell \rangle = \sum_{n = 1}^\infty \frac{1}{2\omega_n} \ee^{-\ii \omega_n(t-t')}  \Psi_n (z) \otimes \overline{\Psi_n (z')}.
\label{G1+1}
\end{align}

The two-point Wightman function \eqref{G1+1} has four tensor-product components. The diagonal ones constitute bulk-bulk and boundary-boundary correlation functions, while the off-diagonal terms define the bulk-boundary correlation functions. In order to define the bulk and boundary renormalized local state polarizations and local Casimir energies that we study in this work in Sec. \ref{Casimir}, we will make use of the bulk-bulk and boundary-boundary two-point Wightman functions, which are given explicitly, in terms of the  eigenfunctions \eqref{Eigenfunctions}, by
\begin{subequations}
\begin{align}\label{BulkWightman}
\langle \Omega_\ell | \hat \Phi^{\rm B}(t,z) \hat \Phi^{\rm B}(t',z') \Omega_\ell \rangle & = \sum_{n = 1}^\infty \frac{1}{2\omega_n} \ee^{-\ii \omega_n(t-t')} \psi_n (z) \overline{\psi_n (z')} \text{ and }  \\\label{BoundaryWightman}
\langle \Omega_\ell | \hat \Phi^{\partial}(t) \hat \Phi^{\partial}(t') \Omega_\ell \rangle & = \sum_{n = 1}^\infty \frac{1}{2\omega_n} \ee^{-\ii \omega_n(t-t')} \left| \beta'_1 \psi_n(\ell) - \beta_2'    \partial_z\psi_n(\ell) \right|^2,
\end{align}
\end{subequations}
respectively.

At a positive  temperature, $T := 1/\beta > 0$, the thermal equilibrium state for our problem can be defined in Fock space by the usual Gibbs formula. See Section 2.3 of \cite{Fulling:1987}
for details. Here we only quote the  formula for the positive-temperature  two-point Wightman function given in eq. (2.43) of \cite{Fulling:1987}.
%
\begin{align}
\langle\hat \Phi(t,z)  \hat \Phi(t',z') \rangle_\beta 
&=\sum_{n = 1}^\infty \frac{1}{2\omega_n} \frac{\Psi_n (z)\otimes \overline{{\Psi_n (z')}}}{1-\ee^{-\beta \omega_n}} \left( \ee^{-\ii \omega_n(t-t')} + \ee^{-\beta \omega_n} \ee^{\ii \omega_n(t-t')} \right).
\label{FiniteTMassive}
\end{align}
The limit $\beta \to \infty$ of  \eqref{FiniteTMassive} yields \eqref{G1+1}, which justifies that the vacuum state can be seen as a zero-temperature state.

In order to define the bulk and boundary renormalized local state polarizations and local Casimir energies at positive temperature, that we will study in Sec. \ref{CasimirTemp}, we will make use of the bulk-bulk and boundary-boundary positive-temperature two-point Wightman functions, which are given explicitly, in terms of the  eigenfunctions \eqref{Eigenfunctions}, by
\begin{subequations}
\begin{align} \langle\hat \Phi^{\rm B}(t,z)  \hat \Phi^{\rm B}(t',z') \rangle_\beta  & = \sum_{n = 1}^\infty \frac{1}{2\omega_n} \frac{\psi_n (z) \overline{\psi_n (z')}}{1-\ee^{-\beta \omega_n}} \left( \ee^{-\ii \omega_n(t-t')} + \ee^{-\beta \omega_n} \ee^{\ii \omega_n(t-t')} \right), \text{ and } \label{BulkWightmanT} \\
\langle\hat \Phi^\partial(t,z)  \hat \Phi^\partial(t',z') \rangle_\beta  & = \sum_{n = 1}^\infty \frac{1}{2\omega_n} \frac{\left| \beta'_1 \psi_n(\ell) - \beta_2' \partial_z\psi_n(\ell) \right|^2}{1-\ee^{-\beta \omega_n}} \left( \ee^{-\ii \omega_n(t-t')} + \ee^{-\beta \omega_n} \ee^{\ii \omega_n(t-t')} \right), \label{BoundaryWightmanT}
\end{align}
\end{subequations}
respectively.


\subsection{The Hadamard property,  the renormalized local state polarization and the local  Casimir energy}
\label{subsec:Hadamard}

The linear system that we study in this work consists of a bulk field, defined on the interval (i.e. in $1+1$ dimensions), and a boundary observable defined on the right-end boundary of the interval (at $z = \ell$). In this section, we will define the renormalized local state polarization and local Casimir  energy for the bulk field and  the boundary observable of the system. For the bulk field, we briefly introduce the notion of Hadamard states.

The {\it physically-relevant} states of linear fields satisfy the so-called Hadamard condition, which characterizes the short-distance behavior of field correlations. In $1+1$ dimensions, which is the case that we study in this paper, we say that a state vector, $\Psi$, is locally Hadamard if its two-point Wightman function satisfies that for any two points, $\x$ and $\x'$ in a convex normal neighborhood of spacetime  $\mathcal{N}$ it holds that  \cite{Hawking:1973uf, Radzikowski:1996pa,Wald:1995yp}
\begin{align}
\langle \Psi | \hat{\Phi}(\x) \hat{\Phi}(\x') \Psi \rangle - H_\lambda(\x, \x') = C^\infty(\mathcal{N} \times \mathcal{N}),
\end{align}
where $H_\lambda$ is the $(1+1)$-dimensional Hadamard bi-distribution \cite{Decanini:2005gt}
\begin{align}
H_\lambda(\x, \x') = \frac{1}{4 \pi} \left[ Q(\x, \x') \ln\left(\sigma(\x, \x')/\lambda^2\right) + W(\x, \x') \right].
\end{align}

Here  we denote  $x:=(t,z), x':=(t',z'),$  $\lambda \in \mathbb{R}$ is an arbitrary scale, $\sigma(\x, \x')$ is the half-squared geodesic distance between the points $\x$ and $\x'$, which is unambiguously defined for $\x$ and $\x'$ inside a convex neighborhood, and $Q$ and $W$ are smooth, symmetric coefficients obtained by solving the so-called Hadamard recursion relations.  Furthermore, $\ln z$ is the principal branch of the logarithm with the argument of $z$ in $(-\pi, \pi).$  

We will be interested below (in Sec. \ref{Casimir}) in studying the problem when the potential is $V(z) = 0$. In this case, in local coordinates, we fix the Hadamard bi-distribution to
\begin{align}
 H_{\rm M}((t,z),(t',z')) & :=  -\frac{1}{4 \pi} \left\{ 2 \gamma + \ln \left[ \frac{m^2}{2} \sigma((t,z),(t',z')) \right] \right. \nonumber \\
&  + \frac{m^2}{2} \sigma((t,z),(t',z'))  \left[  \ln \left(\frac{m^2}{2} \sigma ((t,z),(t',z')) \right) + 2\gamma- 2   \right]  \nonumber \\
& \left. + \frac{m^4}{16 \pi} \sigma^2((t,z),(t',z')) \left[ \ln \left(\frac{m^2}{2} \sigma((t,z),(t',z')) \right)+2 \gamma -3 \right] \right\} \nonumber \\
&  + O\left(\sigma^3((t,z),(t',z') \ln \left( \sigma((t,z),(t',z') \right) \right), 
\label{HM}
\end{align}
where $2 \sigma((t,z),(t',z')) = -(t-t')^2+(z-z')^2$ and $\gamma$ is the Euler number.This choice guarantees that the renormalized state  polarization and renormalized stress-energy tensor in the Minkowski vacuum can be set to zero, in agreement with Wald's fourth axiom for the renormalized stress-energy tensor, see e.g. \cite{Wald:1995yp}. 
The subscript $\rm M$ on the left-hand side of  \eqref{HM} places emphasis on this fact. 
Using this prescription, our definitions \eqref{VacCas} below are equivalent to subtracting the Minkowski vacuum, as has been done originally by Kay in \cite{Kay:1978zr} for studying the local Casimir effect with periodic boundary conditions.

In plain words, as $\x' \to \x$, a Hadamard state vector, $\Psi$, will display a logarithmic singularity in the two-point Wightman function for linear theories defined in two spacetime dimensions. This singularity must be renormalized in order to define the renormalized local state polarization and local Casimir energy in the Hadamard state $\Psi$. More precisely, for a linear scalar field of mass $m$ the renormalized local state polarization and local Casimir energy are defined, respectively, by a point-splitting regularization as
\begin{subequations}
\begin{align}
\langle \Psi | ( \hat \Phi_{\rm ren})^2(\x) \Psi \rangle &:= \lim_{\x' \to \x} \left[\langle \Psi | \hat \Phi(\x) \hat \Phi(\x') \Psi \rangle - H_{\rm M}(\x, \x') \right], \label{VacumPol} \\
\langle \Psi |  \hat H_{\rm ren}(\x) \Psi \rangle &:= \lim_{\x' \to \x} \mathscr{T} \left[\langle \Psi | \hat \Phi(\x) \hat \Phi(\x') \Psi \rangle - H_{\rm M}(\x, \x') \right], \label{CasimirEnergy}
\end{align}
\label{VacCas}
\end{subequations}
where the operator $\mathscr{T}$ takes the form $\mathscr{T} = \frac{1}{2} \left(\partial_t \partial_{t'} + \partial_z \partial_{z'} + m^2 + V(z) \right) $ in coordinates, $x:=(t,z)$, defined on a convex normal neighborhood, $\mathcal{N}$.

Recall that a globally-hyperbolic spacetime  is a spacetime without boundary such that the Cauchy problem is globally well posed. It is known that in globally-hyperbolic spacetimes, vacuum and thermal states (more generally, passive states) satisfy the Hadamard condition \cite{Sahlmann:2000fh}.
Our spacetime it is not globally hyperbolic because it has a boundary. However, as we saw in Section~\ref{sec:Class}, when we impose our boundary condition the Cauchy problem is globally well posed. Furthermore,  as we shall see, the Hadamard condition also holds  in our problem, and  renormalized  bulk observables can be defined from the bulk  two-point Wightman function \eqref{BulkWightman} with prescription \eqref{VacCas}.
For boundary observables, as we shall see in Sec. \ref{Casimir}, no Hadamard subtraction is needed, for the coincidence limit $\x' \to \x$ of the boundary two-point Wightman function \eqref{BoundaryWightman} as the limit will turn out to be regular, as it is to be expected for a system of zero spacial dimension.

\section{The Casimir effect}
\label{Casimir}

In this section, we compute the renormalized local state polarization and the local Casimir energy. Throughout this section we set the potential $V(z) = 0$, 
and we always suppose that the   operator $A$ is positive, that is to say, that all its eigenvalues are positive. In particular, this is true if the assumptions of Proposition~\ref{pos}  hold. We will be interested in the case in which $\beta_2' \neq 0$.


The strategy that we will follow in this section for obtaining the renormalized local state polarization and local Casimir energy is to improve Fulton's estimates in our cases of interest in order to characterize the ultraviolet behavior of the two-point Wightman function. This will allow us in turn to extract and remove the singularity structure of the two-point Wightman function in the coincidence limit and appropriately define the renormalized local  state  polarization and the local   Casimir energy. The required estimates are obtained in Appendix \ref{App:Estimates}.

The case in which a Dirichlet boundary condition is imposed at $z = 0$ is treated in Sec. \ref{subsec:Dirichlet}, while the case in which a Robin (including Neumann) boundary condition is imposed at $z = 0$ appears in Sec. \ref{subsec:Robin}. From now on, we will use labels ${\rm D}$ and ${\rm R}$ to refer to Dirichlet and Robin quantities (states, eigenfunctions and eigenvalues), respectively.  In particular, on spite of the fact that the vacuum state $\Omega_\ell$ does not depend on the boundary condition, we will use the notation $\Omega_\ell^{ ({\rm D })},$   $\Omega_\ell^{ ({\rm R })}$ to emphasize that we consider, respectively,  the  Dirichlet or the Robin boundary condition at $z=0.$ 

\subsection{Dirichlet boundary condition at $z = 0$}
\label{subsec:Dirichlet}

The two-point Wightman function in the Dirichlet case is given by
\begin{align}
\langle \Omega_\ell^{({\rm D})} | \hat \Phi(t,x) \hat \Phi(t',x') \Omega_\ell^{({\rm D})} \rangle = \sum_{n = 1}^\infty \frac{1}{2\omega_n} \ee^{-\ii \omega_n(t-t')} \Psi_n^{({\rm D})} (z) \otimes \Psi_n^{({\rm D})} (z'),
\end{align}
with $(\omega_n^{\rm D})^2 = (s_n^{\rm D})^2 + m^2$, where the normalized eigenfunctions are given by
\begin{align}\label{t.t.2}
\Psi_n^{(\rm D)}(z) & = \begin{pmatrix}
				\psi_n^{({\rm D})}(z) \\ 
				\psi_n^{\partial \, ({\rm D})}
\end{pmatrix} 
 = \mathcal{N}_n^{\rm D} 
	\begin{pmatrix}
          -\sin(s_n^{\rm D} z)  \\
           -\beta_1' \sin(\ell s_n^{\rm D}) + \beta_2' s_n^{\rm D} \cos(\ell s_n^{\rm D})
    \end{pmatrix},  \\
\mathcal{N}_n^{\rm D} & = \left(\frac{L}{2} + \frac{2 \rho (s_n^{\rm D})^2 - (\beta_1' (s_n^{\rm D})^2 + (\beta_1+ \beta_1' m^2))(\beta_2' (s_n^{\rm D})^2 + (\beta_2+ \beta_2' m^2))}{2 \left[ (s_n^{\rm D})^2(\beta_2' (s_n^{\rm D})^2 + (\beta_2+ \beta_2' m^2))^2 + (\beta_1' (s_n^{\rm D})^2 + (\beta_1+ \beta_1' m^2))^2 \right]} \right)^{-1/2}.
\label{DiriEigen}
\end{align}

The eigenfunctions \eqref{DiriEigen} are orthonormal with respect to the inner product on the classical Hilbert space $\mathcal{H}$, i.e.,
\begin{equation}
(\Psi_n^{(\rm D)},\Psi_m^{(\rm D)})_\mathcal{H} = \int_0^\ell \! \dd z \, \psi_n^{({\rm D})}(z) \psi_m^{({\rm D})}(z) + \frac{\psi_n^{\partial \, ({\rm D})} \psi_m^{\partial \, ({\rm D})}}{\rho} = \delta_{nm},
\label{InnerD}
\end{equation}
where $\delta_{nm}$ is the Kronecker delta, with $\rho$ defined in \eqref{w.3}. Recall that the eigenfunctions $\Psi_n^{(\rm D)}$ are real valued. The eigenvalues $(\omega_n^{\rm D})^2 = (s_n^{\rm D})^2 + m^2$ satisfy the asymptotic behavior displayed in  \eqref{snD3} of App. \ref{App:Estimates}.

In order to compute the renormalized local state polarization and the local  Casimir energy in the bulk, we will make use of the bulk-bulk component of the two-point  Wightman function, while for computing the boundary  renormalized local state polarization and boundary  local Casimir energy we will make use of the boundary-boundary component. They are given, respectively, by
\begin{subequations}  
\begin{align}
 \langle \Omega_\ell^{({\rm D})} | \hat \Phi^{\rm B}(t,z) \hat \Phi^{\rm B}(t',z') \Omega_\ell^{({\rm D})} \rangle & = \sum_{n = 1}^\infty \frac{1}{2\omega_n^{\rm D}} \ee^{-\ii \omega_n^{\rm D}(t-t')} \psi_n^{({\rm D})} (z) \psi_n^{({\rm D})} (z'), \\
 \langle \Omega_\ell^{({\rm D})} | \hat \Phi^\partial(t) \hat \Phi^\partial(t') \Omega_\ell^{({\rm D})} \rangle & = \sum_{n = 1}^\infty \frac{1}{2\omega_n^{\rm D}} \ee^{-\ii \omega_n^{\rm D}(t-t')} \psi_n^{\partial \, ({\rm D})} \psi_n^{\partial \, ({\rm D})}.\label{rr.2}
\end{align}
\end{subequations}

In Sec. \ref{subsubsec:DiriState} we obtain the bulk and boundary local renormalized state polarization, while in Sec. \ref{subsubsec:DiriCasimir} we obtain the bulk and boundary local Casimir energy, in the case with a Dirichlet boundary condition at $z = 0$.

\subsubsection{The renormalized local state polarization}
\label{subsubsec:DiriState}

We begin by computing the renormalized local state polarization. Following our discussion in Sec. \ref{subsec:Hadamard}, for the bulk field and the boundary observable, we are interested in computing, respectively,
\begin{subequations}
\label{DPola}
\begin{align}
\langle \Omega_\ell^{({\rm D})} | ( \hat \Phi^{{\rm B}}_{\rm ren})^2(t,z) \Omega_\ell^{({\rm D})} \rangle &:= \lim_{(t',z') \to (t,z)} \left[\langle \Omega_\ell^{({\rm D})} | \hat \Phi^{\rm B}(t,z) \hat \Phi^{\rm B}(t',z') \Omega_\ell^{({\rm D})} \rangle - H_{\rm M}((t,z),(t',z')) \right], \\
\langle \Omega_\ell^{({\rm D})} | ( \hat  \Phi^{\partial}_{\rm ren} )^2(t) \Omega_\ell^{({\rm D})} \rangle &:= \lim_{t' \to t} \langle \Omega_\ell^{({\rm D})} | \hat \Phi^{\partial}(t) \hat \Phi^{\partial}(t') \Omega_\ell^{({\rm D})} \rangle,\label{DPOLAB}
\end{align}
\end{subequations}
where  $H_{\rm M}$ is defined in \eqref{HM}.

We begin by computing the bulk renormalized local state polarization. We seek to control the ultraviolet behavior of
\begin{align}
& \langle \Omega_\ell^{({\rm D})} | \hat \Phi^{\rm B}(t,z) \hat \Phi^{\rm B}(t',z') \Omega_\ell^{({\rm D})} \rangle = \sum_{n = 1}^\infty \frac{(\mathcal{N}_n^{\rm D})^2}{2\omega_n^{\rm D}} \ee^{-\ii \omega_n^{\rm D}(t-t')}  \sin(s_n^{\rm D} z) \sin(s_n^{\rm D} z') \nonumber \\
& = \sum_{n = 1}^\infty \frac{(\mathcal{N}_n^{\rm D})^2}{8\omega_n^{\rm D}} \ee^{-\ii \omega_n^{\rm D}(t-t')}  \left[ \ee^{\ii s_n^{\rm D}(z - z')} + \ee^{-\ii s_n^{\rm D}(z - z')} - \ee^{\ii s_n^{\rm D}(z + z')}  - \ee^{-\ii s_n^{\rm D}(z + z')} \right].
\label{SingSuppD}
\end{align}

It is clear from the asymptotic estimates of $s_n^{\rm D}$ (see App. \ref{App:Estimates}) that in the limit $(t,z) \to (t',z')$ the sum on the right-hand side of  \eqref{SingSuppD} fails to converge for the first two terms inside the square bracket. We isolate the divergent behavior in this limit by noting that the large $n$ behavior of these terms is (using  \eqref{snD3})
\begin{align}
\frac{(\mathcal{N}_n^{\rm D})^2}{8\omega_n} \ee^{-\ii \omega_n^{\rm D}(t-t')} \ee^{\pm \ii s_n^{\rm D}(z \pm z')} & = \ee^{\frac{\ii \pi}{\ell} \left(n  - \frac{1}{2}\right) [-(t-t') \pm (z \pm z')]} \left( \frac{1}{4 \pi n} + O\left(n^{-2}\right) \right).
\label{On-2}
\end{align}

Thus, adding and subtracting terms of the form $\ee^{\frac{\ii \pi}{\ell} \left(n  - \frac{1}{2}\right) [-(t-t') \pm (z \pm z')]}/(4 \pi n)$ to the summand, we can write
\begin{align}
 \langle \Omega_\ell^{({\rm D})} |  \hat \Phi^{\rm B}(t,z) \hat \Phi^{\rm B}(t',z') \Omega_\ell^{({\rm D})} \rangle  & = \sum_{n = 1}^\infty \frac{1}{4 \pi n} \left( \ee^{\frac{\ii \pi}{\ell} \left(n  - \frac{1}{2}\right) [-(t-t') + (z - z')]} + \ee^{\frac{\ii \pi}{\ell} \left(n  - \frac{1}{2}\right) [-(t-t') - (z - z')]} \right. \nonumber \\
&  \left. - \ee^{\frac{\ii \pi}{\ell} \left(n  - \frac{1}{2}\right) [-(t-t') + (z + z')]} - \ee^{\frac{\ii \pi}{\ell} \left(n  - \frac{1}{2}\right) [-(t-t') - (z + z')]} \right) \nonumber \\
& + \sum_{n = 1}^\infty \left[ \frac{(\mathcal{N}_n^{\rm D})^2}{8\omega_n^{\rm D}} \ee^{-\ii \omega_n^{\rm D}(t-t')}  \left( \ee^{\ii s_n^{\rm D}(z - z')} + \ee^{-\ii s_n^{\rm D}(z - z')} - \ee^{\ii s_n^{\rm D}(z + z')}  - \ee^{-\ii s_n^{\rm D}(z + z')} \right) \right. \nonumber \\
& \left. - \frac{1}{4 \pi n} \left( \ee^{\frac{\ii \pi}{\ell} \left(n  - \frac{1}{2}\right) [-(t-t') + (z - z')]} + \ee^{\frac{\ii \pi}{\ell} \left(n  - \frac{1}{2}\right) [-(t-t') - (z - z')]} \right. \right. \nonumber \\
& \left. \left. - \ee^{\frac{\ii \pi}{\ell} \left(n  - \frac{1}{2}\right) [-(t-t') + (z + z')]} - \ee^{\frac{\ii \pi}{\ell} \left(n  - \frac{1}{2}\right) [-(t-t') - (z + z')]} \right) \right].
\end{align}

The first sum can be performed immediately using formula \eqref{sumJonq1} of App. \ref{app:UsefulFormulae}, and we obtain that
\begin{align}
& \langle \Omega_\ell^{({\rm D})} |   \hat \Phi^{\rm B}(t,z) \hat \Phi^{\rm B}(t',z') \Omega_\ell^{({\rm D})} \rangle  \nonumber \\
 & = - \frac{1}{4 \pi} \left[ \ee^{-\frac{\ii \pi}{2 \ell}  [-(t-t') + (z - z')] } \ln \left(1-\ee^{\frac{\ii \pi}{\ell}  [-(t-t') + (z - z')]}\right) \right. \nonumber \\
 & + \ee^{-\frac{\ii \pi}{2 \ell}  [-(t-t') - (z - z')] } \ln \left(1-\ee^{\frac{\ii \pi}{\ell}  [-(t-t') - (z - z')]}\right) - \ee^{-\frac{\ii \pi}{2 \ell}  [-(t-t') + (z + z')] } \ln \left(1-\ee^{\frac{\ii \pi}{\ell}  [-(t-t') + (z + z')]} \right) \nonumber \\
&  \left. - \ee^{-\frac{\ii \pi}{2 \ell}  [-(t-t') - (z + z')] } \ln \left(1-\ee^{\frac{\ii \pi}{\ell}  [-(t-t') - (z + z')]}\right) \right] \nonumber \\
& + \sum_{n = 1}^\infty \left[ \frac{(\mathcal{N}_n^{\rm D})^2}{8\omega_n^{\rm D}} \ee^{-\ii \omega_n^{\rm D}(t-t')}  \left[ \ee^{\ii s_n^{\rm D}(z - z')} + \ee^{-\ii s_n^{\rm D}(z - z')} - \ee^{\ii s_n^{\rm D}(z + z')}  - \ee^{-\ii s_n^{\rm D}(z + z')} \right] \right. \nonumber \\
& \left. - \frac{1}{4 \pi n} \left( \ee^{\frac{\ii \pi}{\ell} \left(n  - \frac{1}{2}\right) [-(t-t') + (z - z')]} + \ee^{\frac{\ii \pi}{\ell} \left(n  - \frac{1}{2}\right) [-(t-t') - (z - z')]} \right. \right. \nonumber \\
& \left. \left. - \ee^{\frac{\ii \pi}{\ell} \left(n  - \frac{1}{2}\right) [-(t-t') + (z + z')]} - \ee^{\frac{\ii \pi}{\ell} \left(n  - \frac{1}{2}\right) [-(t-t') - (z + z')]} \right) \right],
\label{VacDBulk1}
\end{align}
where the sum on the right-hand side is absolutely convergent as $(t',z') \to (t,z)$, since by \eqref{snD3} the summand is $O(n^{-2})$ uniformly in $z, z' \in (0, \ell)$ and $t, t'$ in a compact, say $[-T, T]$. The coincidence limit of the first two terms appearing in closed form in the right-hand side of   \eqref{VacDBulk1} is singular and of Hadamard form. Indeed
\begin{align}
& \lim_{(t',z') \to (t,z)}  \left\{- \frac{1}{4 \pi} \left[ \ee^{-\frac{\ii \pi}{2 \ell}  [-(t-t') + (z - z')] } \ln \left(1-\ee^{\frac{\ii \pi}{\ell}  [-(t-t') + (z - z')]}\right) \right. \right. \nonumber \\
 & \left. \left. + \ee^{-\frac{\ii \pi}{2 \ell}  [-(t-t') - (z - z')] } \ln \left(1-\ee^{\frac{\ii \pi}{\ell}  [-(t-t') - (z - z')]}\right) \right] - H_{\rm M}((t,z),(t',z')) \right\} = \frac{1}{4 \pi} \ln \left( \frac{m^2 \ell^2}{ 4 \pi^2} \right) + \frac{\gamma}{2 \pi}.
\end{align}

For the sums appearing in  \eqref{VacDBulk1}, we can use the dominated convergence theorem to take the limit inside the sums due to the uniform $O(n^{-2})$ behavior of the summand. Thus we finally obtain that,
\begin{align}
\langle \Omega_\ell^{({\rm D})} |  ( \hat \Phi^{{\rm B}}_{\rm ren})^2(t,z) \Omega_\ell^{({\rm D})} \rangle & := \lim_{(t',x') \to (t, x)}\left[\langle \Omega_\ell^{({\rm D})} | \hat \Phi^{\rm B}(t,z) \hat \Phi^{\rm B}(t',z') \Omega_\ell^{({\rm D})} \rangle - H_{\rm M} ( (t,z), (t',z')) \right] \nonumber \\
& =  \frac{1}{4 \pi} \ln \left( \frac{m^2 \ell^2}{4 \pi^2} \right) +   \frac{\gamma}{2 \pi} + \frac{1}{2 \pi}\Re \left( \ee^{\frac{\ii \pi}{\ell} z} \ln \left(1 - \ee^{\frac{\ii 2 \pi}{\ell} z} \right) \right) \nonumber \\
& + \sum_{n = 1}^\infty \left[ \frac{(\mathcal{N}_n^{\rm D})^2}{2\omega_n} \sin^2\left({s_n^{\rm D} z}\right) - \frac{1}{\pi n} \sin^2\left(\frac{\pi}{\ell} (n-1/2) z \right)  \right].
\label{DirichletVacuumBulk}
\end{align}

The sum appearing in  \eqref{DirichletVacuumBulk} is absolutely convergent, has $O(n^{-2})$ summand uniformly in $z$ and is uniformly bounded in $z$. We observe a logarithmic divergence as $z \to 0$ and as $z \to \ell$. These logarithmic divergences are integrable and occur also e.g. for the computation of the renormalized local state polarization and the local Casimir energy in the case of Dirichlet boundary conditions at the ends of  the interval. See e.g. \cite[Chap. 5]{Fulling:1989nb}. We stress that the bulk renormalized local state polarization is time-independent.

We proceed to compute the boundary renormalized local state polarization. We have that
\begin{align}
 \langle \Omega_\ell^{({\rm D})} | \hat \Phi^{\partial}(t) \hat \Phi^{\partial}(t') \Omega_\ell^{({\rm D})} \rangle & =   \sum_{n = 1}^\infty \frac{(\mathcal{N}^{\rm D}_n)^2}{2\omega^{\rm D}_n} \ee^{-\ii \omega^{\rm D}_n(t-t')} \left[-\beta_1' \sin(\ell s_n^{\rm D}) + \beta_2' s_n^{\rm D} \cos(\ell s_n^{\rm D})\right]^2.
\label{DiriPolaBound}
\end{align}

We note that by \eqref{snD3}
\begin{align}
& \left| \frac{(\mathcal{N}^{\rm D}_n)^2}{2\omega^{\rm D}_n} \ee^{-\ii \omega^{\rm D}_n(t-t')} \left[-\beta_1' \sin(\ell s_n^{\rm D}) + \beta_2' s_n^{\rm D} \cos(\ell s_n^{\rm D})\right]^2 \right| \nonumber \\
& = \frac{(\mathcal{N}^{\rm D}_n)^2}{2\omega^{\rm D}_n}  \left[-\beta_1' \sin(\ell s_n^{\rm D}) + \beta_2' s_n^{\rm D} \cos(\ell s_n^{\rm D})\right]^2 \nonumber \\
&  = \frac{\ell^4 (\beta_1 \beta_2'-\beta_1' \beta_2)^2}{\pi ^5 (\beta_2')^2 n^5} + O\left( n^{-6} \right).
\label{Dpolabound}
\end{align}
From  \eqref{Dpolabound} it follows that the limit $t' \to t$ in the definition of the boundary renormalized local state  polarization \eqref{DPOLAB} can be taken inside the sum by a dominated convergence argument, since the bounding function appearing in the right-hand side  \eqref{Dpolabound} is $t, t'$-independent and summable due to the polynomial fall-off as $O(n^{-5})$ when $n \to \infty$. Hence, we have that
\begin{align}
& \langle \Omega_\ell^{({\rm D})} |  ( \hat  \Phi^{\partial}_{\rm ren} )^2 \Omega_\ell^{({\rm D})} \rangle   = \sum_{n = 1}^\infty \frac{(\mathcal{N}^{\rm D}_n)^2}{2\omega^{\rm D}_n} \left[-\beta_1' \sin(\ell s_n^{\rm D}) + \beta_2' s_n^{\rm D} \cos(\ell s_n^{\rm D})\right]^2 .
\label{DirichletVacuumBound}
\end{align}

The sum on the right-hand side of  \eqref{DirichletVacuumBound} is absolutely convergent, the summand is $O(n^{-5})$, and we stress that the renormalized local state polarization in the boundary is time independent.

\subsubsection{The local Casimir energy}
\label{subsubsec:DiriCasimir}


Following our discussion in Sec. \ref{subsec:Hadamard}, for the bulk field and the boundary observable, we are interested in computing
\begin{subequations}
\label{DH}
\begin{align}
\langle \Omega_\ell^{({\rm D})} | \hat H^{{\rm B}}_{\rm ren}(t,z) \Omega_\ell^{({\rm D})} \rangle &= \lim_{(t',z') \to (t,z)} \frac{1}{2}\left[ \left(\partial_t \partial_{t'} + \partial_z \partial_{z'} + m^2 \right)\langle \Omega_\ell^{({\rm D})} | \hat \Phi^{\rm B}(t,z) \hat \Phi^{\rm B}(t',z') \Omega_\ell^{({\rm D})} \rangle \right. \nonumber \\
& \left. - \left(\partial_t \partial_{t'} + \partial_z \partial_{z'}  + m^2 \right) H_{\rm M}((t,z),(t',z')) \right], \label{DHBulk} \\
\langle \Omega_\ell^{({\rm D})} | \hat H^{\partial}_{\rm ren}(t) \Omega_\ell^{({\rm D})} \rangle &= \lim_{t' \to t} \frac{1}{2}\left(\beta_1' \partial_t \partial_{t'}  - \beta_1 \right) \langle \Omega_\ell^{({\rm D})} | \hat \Phi^{\partial}(t) \hat \Phi^{\partial}(t') \Omega_\ell^{({\rm D})} \rangle, \label{DHBound}
\end{align}
\end{subequations}
where the Hadamard bi-distribution, $H_{\rm M}$ is chosen as in \eqref{HM}.

We begin by computing the bulk local Casimir energy. We are interested in controlling the ultraviolet behavior of the sum defining the first term on the right-hand side of \eqref{DHBulk}
\begin{align}
& \frac{1}{2} \left(\partial_t \partial_{t'}  + \partial_z \partial_{z'} + m^2 \right)  \langle \Omega_\ell^{({\rm D})} | \hat \Phi^{\rm B}(t,z) \hat \Phi^{\rm B}(t',z') \Omega_\ell^{({\rm D})} \rangle \nonumber \\
& = \sum_{n = 1}^\infty \frac{(\mathcal{N}_n^{\rm D})^2 }{4 \omega^{\rm D}_n} \ee^{-\ii \omega^{\rm D}_n(t-t')} \left[ \left( (\omega^{\rm D}_n)^2 + m^2 \right) \sin(s_n^{\rm D} z) \sin(s_n^{\rm D} z') + (s_n^{\rm D})^2 \cos(s_n^{\rm D} z) \cos(s_n^{\rm D} z')\right],
\end{align}
which using the dispersion relation, $(\omega^{\rm D}_n)^2 = (s_n^{\rm D})^2+m^2$, can be written as
\begin{subequations}
\label{SingSuppDH}
\begin{align}
& \frac{1}{2} \left(\partial_t \partial_{t'}  + \partial_z \partial_{z'} + m^2 \right)  \langle \Omega_\ell^{({\rm D})} | \hat \Phi^{\rm B}(t,z) \hat \Phi^{\rm B}(t',z') \Omega_\ell^{({\rm D})} \rangle = \sum_{n = 1}^\infty \left[ D_{{\rm B} n}^{{\rm H} (1)}(t,t',z,z')+ D_{{\rm B} n}^{{\rm H} (2)}(t,t',z,z') \right], \\
& D_{{\rm B} n}^{{\rm H} (1)}(t,t',z,z') :=  \frac{(\mathcal{N}_n^{\rm D})^2}{8} \omega^{\rm D}_n \ee^{-\ii \omega^{\rm D}_n(t-t')}  \left[ \ee^{\ii s_n^{\rm D}(z - z')} + \ee^{-\ii s_n^{\rm D}(z - z')} \right], \label{DBnH1} \\
& D_{{\rm B} n}^{{\rm H} (2)}(t,t',z,z') := - \frac{(\mathcal{N}_n^{\rm D})^2 m^2}{8 \omega^{\rm D}_n}  \ee^{-\ii \omega^{\rm D}_n(t-t')}  \left[  \ee^{\ii s_n^{\rm D}(z + z')}  + \ee^{-\ii s_n^{\rm D}(z + z')} \right]. \label{DBnH2}
\end{align}
\end{subequations}

We begin by studying the singular behavior as ${(t',z') \to (t,z)}$ of the term $\sum_{n = 1}^\infty D_{{\rm B} n}^{{\rm H} (1)}(t,t',z,z')$. It follows from the estimates \eqref{snD3} that the right-hand side of  \eqref{DBnH1} fails to converge in the coincidence limit. Thus, as we have done for the renormalized local state polarization, we seek to isolate the singularity structure in this limit. Let
\begin{align}
h^{\pm}(t-t',z-z')  & := \ee^{\ii \frac{\pi}{\ell}\left(n-\frac{1}{2}\right) [-(t-t')\pm(z-z')]}\left\{1 - \ii \left(\omega^{\rm D}_n - \frac{\pi}{\ell}\left(n-\frac{1}{2}\right) \right)(t-t') \pm \ii \left(s_n^{\rm D} - \frac{\pi}{\ell}\left(n-\frac{1}{2}\right) \right) (z-z') \right. \nonumber \\
& \left. \left.  + \frac{1}{2} \left[- \ii \left(\omega^{\rm D}_n - \frac{\pi}{\ell}\left(n-\frac{1}{2}\right) \right)(t-t') \right. \pm \ii \left(s_n^{\rm D} - \frac{\pi}{\ell}\left(n-\frac{1}{2}\right) \right) (z-z') \right]^2 \right\}.
\end{align}

By estimate \ref{LemmaBound}
\begin{align}
& \left| \omega^{\rm D}_n \left(\ee^{-\ii \omega^{\rm D}_n(t-t')} \ee^{\pm \ii s_n^{\rm D} (z-z')}  - h^\pm(t-t',z-z') \right) \right| \nonumber \\
& \hspace{30 pt}  \leq \frac{\omega^{\rm D}_n}{6}\left|- \ii \left(\omega^{\rm D}_n - \frac{\pi}{\ell}\left(n-\frac{1}{2}\right) \right)(t-t') \pm \ii \left(s_n^{\rm D} - \frac{\pi}{\ell}\left(n-\frac{1}{2}\right) \right) (z-z') \right|^3.
\label{DiriH1}
\end{align}

It follows from the asymptotic expansion \eqref{snD3} that the bound on the right-hand side of \eqref{DiriH1} can be further bounded uniformly in $t, t'$ and $z, z'$ for $t, t' \in [-T,T]$ and $z, z' \in (0, \ell)$ by a summable function, which behaves as $O(n^{-2})$ as $n \to \infty$ in an analogous way to the bulk renormalized local state polarization case above, whereby a dominated convergence argument  allows us to write
\begin{align}
\lim_{(t',z') \to (t,z)} \sum_{n = 1}^\infty D_{{\rm B} n}^{{\rm H} (1)}(t,t',z,z') = \lim_{(t',z') \to (t,z)} \sum_{n = 1}^\infty \frac{(\mathcal{N}_n^{\rm D})^2}{8} \omega^{\rm D}_n \left(h^+(t-t',z-z') + h^-(t-t',z-z') \right).
\end{align}

Performing an asymptotic expansion for the summand using \eqref{snD3}, we further write
\begin{subequations}
\begin{align}
& \lim_{(t',z') \to (t,z)}   \sum_{n = 1}^\infty D_{{\rm B} n}^{{\rm H} (1)}(t,t',z,z')  = \lim_{(t',z') \to (t,z)} \sum_{n = 1}^\infty \left[ P_n^+((t-t'),(z-z')) + P_n^-((t-t'),(z-z')) \right]  \nonumber \\
& + \lim_{(t',z') \to (t,z)}  \sum_{n = 1}^\infty \left[ \frac{(\mathcal{N}_n^{\rm D})^2}{8} \omega^{\rm D}_n  \left(h^+(t-t',z-z') + h^-(t-t',z-z') \right)  - P_n^+((t-t'),(z-z')) - P_n^-((t-t'),(z-z')) \right]
\label{DBnH1-2}
\end{align}
with
\begin{align}
& P_n^\pm((t-t'),(z-z')) \nonumber \\
& := \ee^{\frac{\ii \pi}{\ell} \left(n  - \frac{1}{2}\right) [-(t-t') \pm (z - z')]} \left[ \frac{\pi n }{4 \ell^2}  \right. - \frac{2 \ii \beta_1' [-(t-t')\pm(z-z')]+\beta_2' \left(\pi +\ii \ell m^2 (t-t') \right)}{8 \beta_2' \ell^2}  \nonumber \\
& + \left. \frac{(\beta_2')^2 \ell^2 m^2 \left(4-m^2 (t-t')^2\right)-4 \beta_1' \beta_2' \ell m^2 (t-t') [-(t-t')\pm(z-z')]-4 (\beta_1')^2 [-(t-t')\pm(z-z')]^2}{32 \pi  (\beta_2')^2 \ell^2 n}   \right].
\end{align}
\end{subequations}
The second sum on the right-hand side of  \eqref{DBnH1-2} converges for all values of $t,t'$ and $z,z'$ -- one can verify  from the asymptotic estimate \eqref{snD3} that the summand behaves as $O(n^{-2})$ as $n \to \infty$,  for all values of $t, t'$ in a compact set and $z'z' \in (0,\ell)$-- while the first sum can be performed analytically using formulae \eqref{sumJonq1}, \eqref{sumJonq0} and \eqref{sumJonq-1}, and the sum contains a distributional singularity as $(t',z') \to (t,z)$ that compensates the Hadamard singular structure. We have that
\begin{align}
\lim_{(t',z') \to (t,z)} & \left[\sum_{n = 1}^\infty D_{{\rm B} n}^{{\rm H} (1)}(t,t',z,z') - \frac{1}{2} \left( \partial_t \partial_{t'} + \partial_z \partial_{z'} + m^2 \right) H_{\rm M} (t,t',z,z') \right] \nonumber \\
& = \frac{\pi ^2 \beta_2'+6 (2 \gamma -1) \beta_2' \ell^2 m^2+6 \beta_2' \ell^2 m^2 \ln \left(\frac{\ell^2 m^2}{4 \pi ^2}\right)+24 \beta_1' \ell}{48 \pi  \beta_2' \ell^2} + \sum_{n = 1}^\infty  \left( \frac{(\mathcal{N}_n^{\rm D})^2 \omega^{\rm D}_n}{4} - \frac{\pi n }{2 \ell^2} + \frac{\pi }{4  \ell^2}  - \frac{ m^2 }{4 \pi  n} \right) ,
\label{HDiriBulk1}
\end{align}
where the summand on the second term of the right-hand side of  \eqref{HDiriBulk1} is $O(n^{-2})$.

The contribution to the local Casimir energy coming from the term \eqref{DBnH2} can be handled as in the case of the renormalized local state polarization. See the discussion begining at  \eqref{On-2}. We have that
\begin{align}
 \lim_{(t',z')\to(t,z)} \sum_{n = 1}^\infty D_{{\rm B} n}^{{\rm H} (2)}(t,t',z,z') &  = - \lim_{(t',z')\to(t,z)} \sum_{n = 1}^\infty \frac{(\mathcal{N}_n^{\rm D})^2 m^2}{8 \omega^{\rm D}_n}  \ee^{-\ii \omega^{\rm D}_n(t-t')}  \left[  \ee^{\ii s_n^{\rm D}(z + z')}  + \ee^{-\ii s_n^{\rm D}(z + z')} \right] \nonumber \\
& = \frac{m^2}{2 \pi} \Re\left( \ee^{-\ii \frac{\pi}{\ell} z} \ln \left( 1- \ee^{\ii \frac{2 \pi}{\ell} z} \right) \right) + m^2 \sum_{n = 1}^\infty \left[ \frac{1}{4 \pi n } - \frac{(\mathcal{N}_n^{\rm D})^2}{8 \omega^{\rm D}_n} \right. \nonumber \\
& + \left. \frac{(\mathcal{N}_n^{\rm D})^2}{4 \omega^{\rm D}_n} \sin^2\left( s_n^{\rm D} z \right) - \frac{1}{2 \pi n} \sin^2\left( \frac{\pi}{\ell}(n-1/2) z \right) \right],
\label{HDiriBulk2}
\end{align}
where the summand of the second term on right-hand side of  \eqref{HDiriBulk2} is $O(n^{-2})$ for all values of $z\in (0, \ell)$ and the sum converges absolutely and uniformly in $z$.

Adding up  \eqref{HDiriBulk1} and \eqref{HDiriBulk2}, we find that the  bulk local Casimir energy is
\begin{align}
 \langle \Omega_\ell^{({\rm D})} |  \hat H^{{\rm B}}_{\rm ren}(t,z) \Omega_\ell^{({\rm D})} \rangle & = \frac{\pi ^2 \beta_2'+6 (2 \gamma -1) \beta_2' \ell^2 m^2+6 \beta_2' \ell^2 m^2 \ln \left(\frac{\ell^2 m^2}{4 \pi ^2}\right)+24 \beta_1' \ell}{48 \pi  \beta_2' \ell^2} \nonumber \\
 & + \frac{m^2}{2 \pi} \Re\left( \ee^{-\ii \frac{\pi}{\ell} z} \ln \left( 1- \ee^{\ii \frac{2 \pi}{\ell} z} \right) \right) + \sum_{n = 1}^\infty \left[  \left( \frac{(\mathcal{N}_n^{\rm D})^2 \omega^{\rm D}_n}{4} - \frac{\pi n }{2 \ell^2} + \frac{\pi }{4  \ell^2}  - \frac{ m^2 }{4 \pi  n} \right) \right. \nonumber \\
&  \left.  - m^2 \left( \frac{(\mathcal{N}_n^{\rm D})^2 }{8 \omega^{\rm D}_n} - \frac{1}{4 \pi n} \right) + \frac{(\mathcal{N}_n^{\rm D})^2 m^2}{4 \omega^{\rm D}_n} \sin^2\left( s_n^{\rm D} z \right) - \frac{m^2}{2 \pi n} \sin^2\left( \frac{\pi}{\ell}(n-1/2) z \right) \right].
\label{HDiriBulk}
\end{align}

The sum appearing on the right-hand side of  \eqref{HDiriBulk} is absolutely convergent and uniformly convergent in $z \in (0, \ell)$, the summand is $O(n^{-2})$ for all $z \in (0, \ell)$. We stress that the bulk local Casimir energy is time-independent, hence conserved. We observe a logarithmic divergence as $z \to 0$ and as $z \to \ell$, as is the case for the bulk renormalized local state polarization, see  \eqref{DirichletVacuumBulk}. Since the logarithmic divergence of  \eqref{HDiriBulk} is integrable, the total Casimir energy is finite.

We now focus our attention on the boundary local Casimir energy. A look at  \eqref{t.t.2}, \eqref{rr.2}, and  \eqref{DHBound} indicates that we need to verify the convergence of the summand factor
\begin{align}
\frac{(\mathcal{N}_n^{\rm D})^2 \omega^{\rm D}_n}{4} \left( -\beta_1' \sin(\ell s_n^{\rm D}) + \beta_2' s_n^{\rm D} \cos (\ell s_n^{\rm D}) \right)^2 & =
 \frac{\ell^2 (\beta_1 \beta_2' -\beta_1' \beta_2)^2}{2 \pi^3 (\beta_2')^2 n^3} +  O (n^{-4}),
\label{DiriHBound1}
\end{align}
where the right-hand side of  \eqref{DiriHBound1} is obtained using  \eqref{snD3}. It follows from the $O(n^{-3})$ behavior that the sum defining the boundary local Casimir energy converges absolutely for all values of $t'$ and $t$, and that the limit $t' \to t$ can be taken inside the sum by dominated convergence. We have that 
\begin{align}
 \langle \Omega_\ell^{({\rm D})} | \hat H^{\partial}_{\rm ren}(t) \Omega_\ell^{({\rm D})} \rangle & =   \sum_{n = 1}^\infty \frac{(\mathcal{N}^{\rm D}_n)^2\left( \beta_1' ( \omega_n^{\rm D})^2  - \beta_1 \right)}{4\omega_n^{\rm D}}  \left[-\beta_1' \sin(\ell s_n^{\rm D}) + \beta_2' s_n^{\rm D} \cos(\ell s_n^{\rm D})\right]^2.
\label{DiriHBound}
\end{align}

The summand on the right-hand side of  \eqref{DiriHBound} is $O(n^{-3})$ and the sum is absolutely convergent. The boundary Casimir energy is time-independent.

\subsection{Robin boundary condition at $z = 0$}
\label{subsec:Robin}

In the Robin case, the two-point Wightman function is given by
\begin{align}
\langle \Omega_\ell^{({\rm R})} | \hat \Phi(t,x) \hat \Phi(t',x') \Omega_\ell^{({\rm R})} \rangle = \sum_{n = 1}^\infty \frac{1}{2\omega^{\rm R}_n} \ee^{-\ii \omega^{\rm R}_n(t-t')} \Psi_n^{({\rm R})} (z) \otimes  \Psi_n^{({\rm R})} (z'),
\end{align}
with $(\omega^{\rm R}_n)^2 = (s_n^{\rm R})^2 + m^2$, and with normalized eigenfunctions
\begin{align}
& \Psi_n^{(\rm R)}(z)  = \begin{pmatrix}
				\psi_n^{({\rm R})}(z) \\ 
				\psi_n^{\partial \, ({\rm R})}
\end{pmatrix} 
\nonumber \\      \label{qqq.1}
& = \mathcal{N}_n^{\rm R} 
	\begin{pmatrix}
       \ds   - \frac{\cos \alpha \sin (s_n^{\rm R} z)}{s_n^{\rm R}} +  \sin \alpha \cos( s_n^{\rm R} z)  \\ 
      \ds     \beta_1' \left(\sin \alpha  \cos (\ell s_n^{\rm R})-\frac{\cos \alpha  \sin (\ell s_n^{\rm R})}{s_n^{\rm R}}\right)+\beta_2' (\cos \alpha \cos (\ell s_n^{\rm R})+s_n^{\rm R} \sin \alpha  \sin (\ell s_n^{\rm R}))
    \end{pmatrix},  \\
& \mathcal{N}_n^{\rm R}  := \left\{ -\frac{\sin (2 \alpha )}{4 (s_n^{\rm R})^2} - \left[4 (s_n^{\rm R})^2\Big((\beta_1+ \beta_1' m^2)^2+9(\beta_1')^2 (s_n^{\rm R})^4+2 (\beta_1 + \beta_1' m^2) \beta_1' (s_n^{\rm R})^2 \right. \right. \nonumber \\
& \left. \left. +(s_n^{\rm R})^2 \left((\beta_2 + \beta_2' m^2)+\beta_2' (s_n^{\rm R})^2\right)^2\right) \right]^{-1} \left(\left((s_n^{\rm R})^2-1\right) \cos (2 \alpha )-(s_n^{\rm R})^2-1\right) \nonumber \\
    &  \times \left[\ell \left((\beta_1+\beta_1' m^2)^2+\beta_1'^2 (s_n^{\rm R})^4+2 (\beta_1 + \beta_1' m^2) \beta_1' (s_n^{\rm R})^2+(s_n^{\rm R})^2 \left((\beta_2+\beta_2' m^2) +\beta_2' (s_n^{\rm R})^2\right)^2\right) \right.  \nonumber \\
    & \left. \left.  -(\beta_1+\beta_1' m^2) \left((\beta_2+\beta_2' m^2)+3 \beta_2' (s_n^{\rm R})^2\right)+\beta_1' (s_n^{\rm R})^2 \left((\beta_2+\beta_2' m^2) -\beta_2' (s_n^{\rm R})^2\right)\right] \right\}^{-1/2},
\label{RobinEigen}
\end{align}
which are orthonormal with respect to the inner product on the classical Hilbert space $\mathcal{H}$, i.e.,
\begin{equation}
(\Psi_n^{(\rm R)},\Psi_m^{(\rm R)})_\mathcal{H} = \int_0^\ell \! \dd z \, \psi_n^{({\rm R})}(z) \psi_m^{({\rm R})}(z) + \frac{\psi_n^{\partial \, ({\rm R})} \psi_m^{\partial \, ({\rm R})}}{\rho} = \delta_{nm},
\label{InnerR}
\end{equation}

where $\rho$ is defined in \eqref{w.3}. Note that the eigenfunctions $\Psi_n^{(\rm R)}$ are real valued. The eigenvalues $(\omega_n^{\rm R})^2 = (s_n^{\rm R})^2 + m^2$ obey the asymptotic behavior presented in  \eqref{snR3}. See App. \ref{App:Estimates} for details.

As we have explained in Sec. \ref{sec:QFT}, the bulk renormalized local state polarization and local  Casimir energy are obtained by using the bulk-bulk component of the two-point Wightman function, and for the boundary counterparts we use the boundary-boundary component. They are, respectively,
\begin{subequations}
\begin{align}
 \langle \Omega_\ell^{({\rm R})} | \hat \Phi^{\rm B}(t,z) \hat \Phi^{\rm B}(t',z') \Omega_\ell^{({\rm R})} \rangle & = \sum_{n = 1}^\infty \frac{1}{2\omega_n^{\rm R}} \ee^{-\ii \omega_n^{\rm R}(t-t')} \psi_n^{({\rm R})} (z) \psi_n^{({\rm R})} (z'), \\
 \langle \Omega_\ell^{({\rm R})} | \hat \Phi^\partial(t) \hat \Phi^\partial(t') \Omega_\ell^{({\rm R})} \rangle & = \sum_{n = 1}^\infty \frac{1}{2\omega_n^{\rm R}} \ee^{-\ii \omega_n^{\rm R}(t-t')} \psi_n^{\partial \, ({\rm R})} \psi_n^{\partial \, ({\rm R})}.
\end{align}
\end{subequations}

\subsubsection{The renormalized local state polarization}

For the bulk renormalized local state polarization, we are interested in computing
\begin{align}
\langle \Omega_\ell^{({\rm R})} | ( \hat \Phi^{{\rm B}}_{\rm ren})^2(t,z) \Omega_\ell^{({\rm R})} \rangle = \lim_{(t', z') \to (t,z)}   \left[ \langle \Omega_\ell^{({\rm R})} | \hat \Phi^{\rm B}(t,z) \hat \Phi^{\rm B}(t',z') \Omega_\ell^{({\rm R})} \rangle - H_{\rm M}((t,z),(t',z')) \right].
\label{BPolRob1}
\end{align}

The strategy to obtain the limit on the right-hand side of  \eqref{BPolRob1} is similar to the one in the Dirichlet case. We show the details of the calculation in App. \ref{app:RobVB}, and state directly the result,
\begin{align}
& \langle \Omega_\ell^{({\rm R})} |  ( \hat \Phi^{{\rm B}}_{\rm ren})^2(t,z) \Omega_\ell^{({\rm R})} \rangle 
 =  \frac{1}{4 \pi} \ln \left( \frac{m^2 \ell^2}{4 \pi^2} \right) +   \frac{\gamma}{2 \pi} - \frac{1}{2 \pi}\Re \left( \ee^{-\frac{\ii \pi}{\ell} z} \ln \left(1 - \ee^{\frac{\ii 2 \pi}{\ell} z} \right) \right) \nonumber \\
& + \sum_{n = 1}^\infty \left\{ \left[ \frac{(\mathcal{N}^{\rm R}_n)^2}{2 \omega_n^{\rm R}} \left(\sin^2 \alpha \cos^2 \left( s_n^{\rm R} z\right) + \frac{\cos^2 \alpha}{(s_n^{\rm R})^2} \sin^2 \left( s_n^{\rm R} z\right) \right)  - \frac{1}{ \pi n} \cos^2 \left( \frac{\pi}{\ell}(n-1) z\right) \right]   \right. \nonumber \\
& \left. - \frac{(\mathcal{N}^{\rm R}_n)^2}{2 \omega_n^{\rm R} s_n^{\rm R}}\sin \alpha \cos \alpha      \sin\left(2 s_n^{\rm R} z \right)\right\}.
\label{RobinVacuumBulk}
\end{align}
As is the case for a Dirichlet boundary condition at $z = 0$, the renormalized local state polarization contains a logarithmic divergence at $z = 0$ and $z = \ell$. We stress again that this divergence is integrable, and hence contributes as a finite term to the integrated total bulk state polarization. We also stress that the summand appearing on the right-hand side of  \eqref{RobinVacuumBulk} is $O(n^{-2})$ for all $z \in (0, \ell)$, hence summable. The sum converges absolutely and uniformly in $z$. We also point out that the  bulk renormalized local state  polarization is time-independent.

For the boundary renormalized local state polarization, we wish to compute
\begin{align}
\langle \Omega_\ell^{({\rm R})} | ( \hat  \Phi^{\partial} _{\rm ren})^2(t) \Omega_\ell^{({\rm R})} \rangle = \lim_{t'\to t}    \langle \Omega_\ell^{({\rm R})} | \hat \Phi^\partial(t) \hat \Phi^\partial(t') \Omega_\ell^{({\rm R})} \rangle .
\label{bPolRob1}
\end{align}

The details of the calculations are presented in App. \ref{app:RobVb} and we simply state the result,
\begin{align}
 \langle \Omega_\ell^{({\rm R})} | ( \hat  \Phi^{\partial}_{\rm ren} )^2(t) \Omega_\ell^{({\rm R})} \rangle & = \sum_{n = 1}^\infty \frac{(\mathcal{N}_n^{\rm R})^2}{2\omega_n^{\rm R}} \left[\beta_1' \left(\sin \alpha  \cos (\ell s_n^{\rm R})-\frac{\cos \alpha  \sin (\ell s_n^{\rm R})}{s_n^{\rm R}}\right) \right. \nonumber \\
& \left. +\beta_2' (\cos \alpha \cos (\ell s_n^{\rm R})+s_n^{\rm R} \sin \alpha  \sin (\ell s_n^{\rm R})) \right]^2.
\label{RobinVacuumBoundary}
\end{align}

We note that the boundary renormalized local  state polarization is time-independent. The summand on the right-hand side of  \eqref{RobinVacuumBoundary} is $O(n^{-5})$ and the sum converges absolutely.

\subsubsection{The  local Casimir  energy}

The bulk local Casimir  energy for a Robin boundary condition at $z = 0$ is given by
\begin{align}
& \langle  \Omega_\ell^{({\rm R})} |  \hat H^{\rm B}(t,z) \Omega_\ell^{({\rm R})} \rangle = -\frac{\pi }{24 \ell^2}-\frac{\cot \alpha }{2 \pi  \ell}-\frac{\beta_1'}{2 \pi  \beta_2' \ell}+\frac{m^2}{8 \pi }\left[ 1  + \ln \left( \frac{m^2 \ell^2}{4 \pi^2} \right) \right] + \frac{\gamma m^2 }{4 \pi} \nonumber \\
& - \frac{m^2}{2 \pi} \Re \left[\ee^{-\ii \frac{ 2 \pi }{\ell} z} \ln \left(1-\ee^{\ii  \frac{2 \pi}{\ell} z}\right) \right] - m^2 \sum_{n = 1}^\infty \frac{(\mathcal{N}^{\rm R}_n)^2}{2 \omega_n^{\rm R} s_n^{\rm R}}\sin \alpha \cos \alpha      \sin\left(2 s_n^{\rm R} z \right) \nonumber \\
&+\sum_{n = 1}^\infty \left\{  \frac{(\mathcal{N}^{\rm R}_n)^2}{8 \omega_n^{\rm R}} \left((\omega_n^{\rm R})^2 + (s_n^{\rm R})^2 \right)     \left(\sin^2 \alpha + \frac{\cos^2 \alpha}{(s_n^{\rm R})^2}  \right) - \frac{\pi (n-1)}{2 \ell^2}  \right\} \nonumber \\
& + \frac{m^2}{2}\sum_{n = 1}^\infty \left[ \frac{(\mathcal{N}^{\rm R}_n)^2}{4 \omega_n^{\rm R}} \left(\sin^2 \alpha - (s_n^{\rm R})^{-2} \cos^2 \alpha \right) \cos(2 s_n^{\rm R} z)   - \frac{1}{ 2 \pi n} \cos \left( \frac{2(n-1) \pi}{\ell} z \right) \right]  \nonumber \\
& + \frac{m^2}{2} \sum_{n = 1}^\infty  \left[ \frac{(\mathcal{N}^{\rm R}_n)^2}{2 \omega_n^{\rm R}} \left(\sin^2 \alpha \cos^2 \left( s_n^{\rm R} z\right) + \frac{\cos^2 \alpha}{(s_n^{\rm R})^2} \sin^2 \left( s_n^{\rm R} z\right) \right) - \frac{1}{ \pi n} \cos^2 \left( \frac{\pi}{\ell}(n-1) z\right) \right] .
\label{RobinHBulk}
\end{align}

The computations leading to  \eqref{RobinHBulk} are presented in App. \ref{app:RobHB}. As in the case of the renormalized local state polarization, appearing in \eqref{RobinVacuumBulk}, an integrable, logarithmic divergence appears at $z = 0$ and $z = \ell$, and the result is time independent. The summand defined on the right-hand side of  \eqref{RobinHBulk} is $O(n^{-2})$ for all $z \in (0, \ell)$ and the sum converges absolutely and uniformly in $z$.

In App. \ref{app:RobHb} we  show that the boundary local Casimir energy for a Robin boundary condition at $z = 0$ is given by
\begin{align}
\langle \Omega_\ell^{({\rm R})} | \hat H^\partial_{\rm ren}(t) \Omega_\ell^{({\rm R})} \rangle & = \sum_{n = 1}^\infty \frac{\left[\beta_1' (\omega_n^{\rm R})^2 - \beta_1 \right] (\mathcal{N}_n^{\rm R})^2}{4\omega_n^{\rm R}} \left[\beta_1' \left(\sin \alpha  \cos (\ell s_n^{\rm R})-\frac{\cos \alpha  \sin (\ell s_n^{\rm R})}{s_n^{\rm R}}\right) \right. \nonumber \\
& \left.  +\beta_2' (\cos \alpha \cos (\ell s_n^{\rm R})+s_n^{\rm R} \sin \alpha  \sin (\ell s_n^{\rm R})) \right]^2.
\label{RobinHBoundary}
\end{align}

The summand on the right-hand side of  \eqref{RobinHBoundary} is $O(n^{-3})$ and the sum converges absolutely. We stress that the result for the boundary local Casimir energy is time independent.

This result concludes the study of the system at zero temperature.

\section{The Casimir effect at positive temperature}
\label{CasimirTemp}

In this section we show how to obtain the Casimir effect at positive temperatures. We set the potential $V(z) = 0$, and we always suppose that the   operator $A$ is positive, that is to say, that all its eigenvalues are positive. In particular, this is true if the assumptions of Proposition~\ref{pos}  hold. We are concerned in the case in which $\beta_2' \neq 0$.
We shall follow a strategy that will allow us to write down both the renormalized local state polarization and the local Casimir energy at positive temperature $T = 1/\beta$ in terms of the  quantities at zero temperature. We will deal with the Dirichlet and Robin cases simultaneously.

We write  \eqref{FiniteTMassive} as
\begin{align}\label{we.99}
\langle\hat \Phi(t,z) \hat \Phi(t',z') \rangle_\beta  & = \sum_{n = 1}^\infty \frac{1}{2 \omega_n} \Psi_n (z) \otimes  \overline{{\Psi_n (z')}} \left[ \ee^{-\ii \omega_n(t-t')} + \frac{1}{\ee^{\beta \omega_j} - 1}\left(\ee^{-\ii \omega_n(t-t')} + \ee^{\ii \omega_n(t-t')} 
\right) \right].
\end{align}

By \eqref{G1+1} and \eqref{we.99} we have,
\begin{align}\label{we.100}
\langle\hat \Phi(t,z)  \hat \Phi(t',z') \rangle_\beta &=  \langle \Omega_\ell | \hat \Phi(t,z) \hat \Phi(t',z') \Omega_\ell \rangle+ \\ \nonumber
& \sum_{n = 1}^\infty \,\frac{1}{2 \omega_n} \,\Psi_n (z)  \otimes \overline{{\Psi_n (z')}}\, \,\frac{1}{\ee^{\beta \omega_j} - 1}\,\left(\ee^{-\ii \omega_n(t-t')} + \ee^{\ii \omega_n(t-t')} 
\right).
\end{align}
We use the following result to write the renormalized local state polarization and local Casimir energy at positive temperature in terms of the quantities at zero temperature.


\begin{lemma}
\label{Lem:term}
Let $\Psi_n$, $s_n$ and $\omega_n$ be as in the cases of Dirichlet or Robin boundary conditions at $z = 0$, $\Psi_n^{({\rm D})}$ ( \eqref{t.t.2}, \eqref{DiriEigen}), respectively $s_n^{\rm D}$ ( . \eqref{snD3}) and $\omega_n^{\rm D} = \left[(s_n^{\rm D})^2 + m^2 \right]^{1/2}$ or $\Psi_n^{({\rm R})}$ ( \eqref{qqq.1}, \eqref{RobinEigen}), $s_n^{\rm R}$ ( \eqref{snR3}) and $\omega_n^{\rm R}= \left[(s_n^{\rm R})^2 + m^2 \right]^{1/2}$. The sums defined by the tensor-product components
\begin{subequations}
\begin{align}
S_{\rm B} &:= \sum_{n = 1}^\infty \frac{1}{2 \omega_n} \psi_n (z) {\psi_n (z')}
 \frac{1}{\ee^{\beta \omega_j} - 1} \left(\ee^{-\ii \omega_n(t-t')} + \ee^{\ii \omega_n(t-t')} \right), \nonumber \\
 S_{\rm \partial} &:= \sum_{n = 1}^\infty \frac{1}{2 \omega_n} \psi_n^{\partial} {\psi_n^\partial}
 \frac{1}{\ee^{\beta \omega_j} - 1} \left(\ee^{-\ii \omega_n(t-t')} + \ee^{\ii \omega_n(t-t')} \right),
\end{align}
\end{subequations}
converge absolutely for $\beta > 0,$ and uniformly for any values of $0 \leq  z, z' \leq  \ell$ and $t, t' \in \mathbb{R}$.
\end{lemma}
\begin{proof}
It follows from  the eigenfunctions $\psi_n$ in the Dirichlet and Robin cases and   from the large $n$ estimates for $s_n$  (see  \eqref{t.t.2}, \eqref{DiriEigen}, \eqref{snD3}, \eqref{Dpolabound}, and \eqref{qqq.1}, \eqref{RobinEigen}, \eqref{snR3}, and \eqref{RobVb1},
  respectively)  that there exists polynomially bounded functions $P^{\rm B}, P^\partial: \mathbb{N} \to \mathbb{R}^+ $ such that for  $z, z' \in (0,\ell),$
\begin{subequations}
\begin{align}
\left| \frac{1}{2 \omega_n} \psi_n (z) {\psi_n (z')} \right| \leq P^{{\rm B}}(n) = O(n^{-1}), \label{PBO1} \\
\left| \frac{1}{2 \omega_n} \left(\psi_n^{\partial}\right)^2 \right| \leq P^{\partial}(n) = O(n^{-5}), \label{Pb0-5}
\end{align}
\end{subequations}
where the right-hand side of  \eqref{PBO1} is uniform for $z, z' \in (0,\ell)$.

It then follows that
\begin{subequations}
\begin{align}
\sum_{n = 1}^\infty \left| \frac{1}{2 \omega_n} \psi_n (z) \psi_n(z')
 \frac{1}{\ee^{\beta \omega_j} - 1} \left(\ee^{-\ii \omega_n(t-t')} + \ee^{\ii \omega_n(t-t')} \right) \right| \leq \sum_{n = 1}^\infty  \frac{2 P^{{\rm B}}(n)}{\ee^{\beta \omega_j} - 1} < \infty, \\
 \sum_{n = 1}^\infty \left| \frac{1}{2 \omega_n} \left( \psi_n^{\partial}\right)^2
 \frac{1}{\ee^{\beta \omega_j} - 1} \left(\ee^{-\ii \omega_n(t-t')} + \ee^{\ii \omega_n(t-t')} \right) \right| \leq \sum_{n = 1}^\infty  \frac{2 P^{\partial}(n)}{\ee^{\beta \omega_j} - 1} < \infty.
\end{align}
\end{subequations}
\end{proof}
We define the renormalized local state polarization  at positive temperature as follows,
\begin{subequations}
\label{DPolaT}
\begin{align}
 \langle(\hat \Phi^{\rm B} )^2(t,z) \rangle_{\beta, {\rm ren} }:= & \lim_{(t',z') \to (t,z)} \left[ \langle\hat \Phi^{\rm B}(t,z)  \hat \Phi^{\rm B}(t',z') \rangle_\beta   - H_{\rm M}((t,z),(t',z')) \right], \\
\langle( \hat \Phi^{\partial} )^2(t) \rangle_{\beta, {\rm ren}} &:= \lim_{t' \to t}  \langle\hat \Phi^{\partial}(t)  \hat \Phi^{\partial}(t') \rangle_\beta .
\end{align}
\end{subequations}

It follows from \eqref{we.100} and Lemma \ref{Lem:term} that the bulk and boundary renormalized local state polarizations at positive temperature are given by
\begin{subequations}
\label{FiniteTpolarization}
\begin{align}
\langle(\hat \Phi^{\rm B} )^2(t,z) \rangle_{\beta, {\rm ren} } & = \langle \Omega_\ell^{{\rm (D/R)}} | ( \hat \Phi^{{\rm B}}_{\rm ren})^2(t, z) \Omega_\ell^{{\rm (D/R)}} \rangle + \sum_{n = 1}^\infty  \frac{\left[\psi_n^{\rm (D/R)} (z) \right]^2}{ \omega_n^{\rm (D/R)} \left(\ee^{\beta \omega_n^{\rm (D/R)}}-1\right)}, \label{TBulkPola}\\
\langle( \hat \Phi^{\partial} )^2(t) \rangle_{\beta, {\rm ren}} & = \langle \Omega_\ell^{{\rm (D/R)}} | ( \hat  \Phi^{\partial} )^2_{\rm ren}(t) \Omega_\ell^{{\rm (D/R)}} \rangle + \sum_{n = 1}^\infty  \frac{\left[\psi_n^{\partial \, \rm (D/R)} \right]^2 }{ \omega_n^{\rm (D/R)} \left(\ee^{\beta \omega_n^{\rm (D/R)}}-1\right)}. \label{TBoundPola}
\end{align}
\end{subequations}

The sum on the right-hand side of  \eqref{TBulkPola} is absolutely convergent and converges exponentially fast for all $z \in (0, \ell)$. The sum on the right-hand side of \eqref{TBoundPola} is absolutely convergent and converges exponentially fast.

We define  the local Casimir energy  at positive temperature in the bulk and in the boundary  as follows,
\begin{subequations}
\begin{align}
\langle \hat H^{{\rm B}}_{\rm ren}(t,z) \rangle_\beta &:= \lim_{(t',z') \to (t,z)} \frac{1}{2}\left[ \left(\partial_t \partial_{t'} + \partial_z \partial_{z'} + m^2 \right) \langle\hat \Phi^{\rm B}(t,z)  \hat \Phi^{\rm B}(t',z') \rangle_\beta \right. \nonumber \\
& \left. - \left(\partial_t \partial_{t'} + \partial_z \partial_{z'}  + m^2 \right) H_{\rm M}((t,z),(t',z')) \right], \label{we.101} \\ \label{we.102}
\langle  \hat H^{\partial}_{\rm ren}(t)\rangle_\beta &= \lim_{t' \to t} \frac{1}{2}\left(\beta_1'\partial_t \partial_{t'} -\beta_1 \right) \langle\hat \Phi^{\partial}(t,z)  \hat \Phi^{\partial}(t',z') \rangle_\beta.
\end{align}
\end{subequations}
Then, by \eqref{we.100}, \eqref{we.101}, and \eqref{we.102} we have, 

\begin{subequations}
\label{FiniteTEnergy}
\begin{align}
 \langle& \hat H^{{\rm B}}_{\rm ren}(t, z) \rangle_\beta  = \langle \Omega_\ell^{\rm (D/R)} | \hat H^{{\rm B}}_{\rm ren}(t, z) \Omega_\ell^{{\rm (D/R)}} \rangle \nonumber \\
& + \sum_{n = 1}^\infty \left[\frac{\left( (\omega_n^{\rm (D/R)})^2 + m^2\right)}{2 \omega_n^{\rm (D/R)}} \frac{\left[\psi_n^{\rm (D/R)} (z)\right]^2}{\ee^{\beta \omega_n^{\rm (D/R)}}-1} + \frac{1}{2 \omega_n^{\rm (D/R)}} \frac{\left[{\partial_z\psi_n^{\rm (D/R)}} (z)\right]^2}{\ee^{\beta \omega_n^{\rm (D/R)}}-1} \right], \label{HTbulk} \\
& \langle  \hat H^{\partial}_{\rm ren}(t) \rangle_\beta  = \langle \Omega_\ell^{{\rm (D/R)}} | \hat H^{\partial}_{\rm ren}(t) \Omega_\ell^{{\rm (D/R)}} \rangle + \sum_{n = 1}^\infty \frac{\left( \beta_1' (\omega_n^{\rm (D/R)})^2 -\beta_1 \right)}{2 \omega_n^{\rm (D/R)}} \frac{\left[\psi_n^{\partial \, \rm (D/R)}\right]^2 }{\ee^{\beta \omega_n^{\rm (D/R)}}-1}, \label{HTboundary}
\end{align}
\end{subequations}
where the sums appearing on the right-hand side of  \eqref{HTbulk} and \eqref{HTboundary} are absolutely convergent by an adaptation of Lemma \ref{Lem:term} and converge exponentially fast. In particular, it can be verified that  polynomially bounded functions such as the ones appearing on  \eqref{PBO1} and \eqref{Pb0-5} can be found for the sums appearing on the right-hand side of  \eqref{HTbulk} and \eqref{HTboundary}. Furthermore, the sum on the right-hand side of  \eqref{HTbulk} converges uniformly in $z$.

We note that, as expected, when $\beta \to \infty$, the results obtained in this section for the renormalized local state polarization and local Casimir energy reduce to the zero temperature ones.

\section{Numerical examples}
\label{sec:Numerics}

In this section we present some numerical examples for the local Casimir energy at temperature zero in the case in which we set a Dirichlet boundary condition at $z = 0$. We fix the parameters of the problem to $\beta_1 = -1$, $\beta_1' = 1$, $\beta_2 = 1$, $m^2 = 1$ and $\ell = 1$, and compare three representative cases, for $\beta_2' = -0.5$, $\beta_2' = -0.05$ and $\beta_2' = 0.5$ in Fig. \ref{fig:beta2p0.5}, \ref{fig:beta2p0.05} and \ref{fig:beta2p-0.5} respectively, in which the total integrated Casimir energy can be positive or negative. While the case $\beta_2' = 0.5$ does not satisfy the hypotheses of Prop. \ref{pos}, our numerics verify that the operator $A$ is positive in this case too. In each case, as $z \to 0$ we observe a negative logarithmic divergence and as $z \to \ell = 1$ a positive logarithmic divergence for the local Casimir energy.

The plots presented in this section are obtained by numerically approximating the local Casimir energy using the first 50 eigenvalues in each case. More precise numerical results can be carried out in cases of interest by computing a larger number of eigenvalues.

We will denote by
\begin{align}
E_0(t) := \int_{0}^\ell \! \dd z \, \langle \Omega_\ell^{({\rm D})} | & \hat H^{{\rm B}}_{\rm ren}(t,z) \Omega_\ell^{({\rm D})} \rangle
\end{align}
the integrated Casimir energy at zero temperature in the bulk, i.e., the bulk total energy of the ground state. Note that in our cases, due to the conservation of the local energy in time, $E_0$ is time-independent.

\begin{figure}[!ht]
\centering
\includegraphics[width=0.45\textwidth]{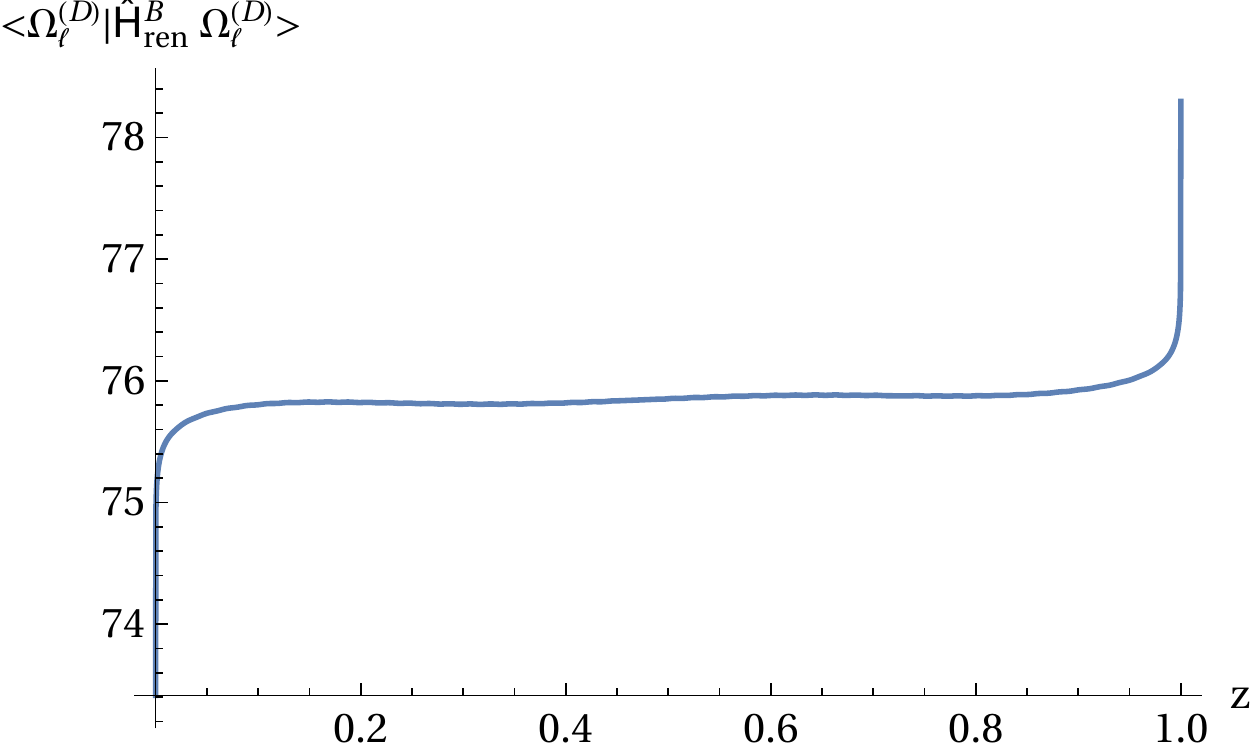}
\caption{Case $\beta_2' = -0.5$. The boundary Casimir energy is $\langle \Omega_\ell^{({\rm D})} | \hat H^{\partial}_{\rm ren} \Omega_\ell^{({\rm D})} \rangle \approx 0.23$. The integrated bulk Casimir energy is $E_0 \approx 75.85$.}
\label{fig:beta2p0.5}
\end{figure}

\begin{figure}[!ht]
\centering
\includegraphics[width=0.45\textwidth]{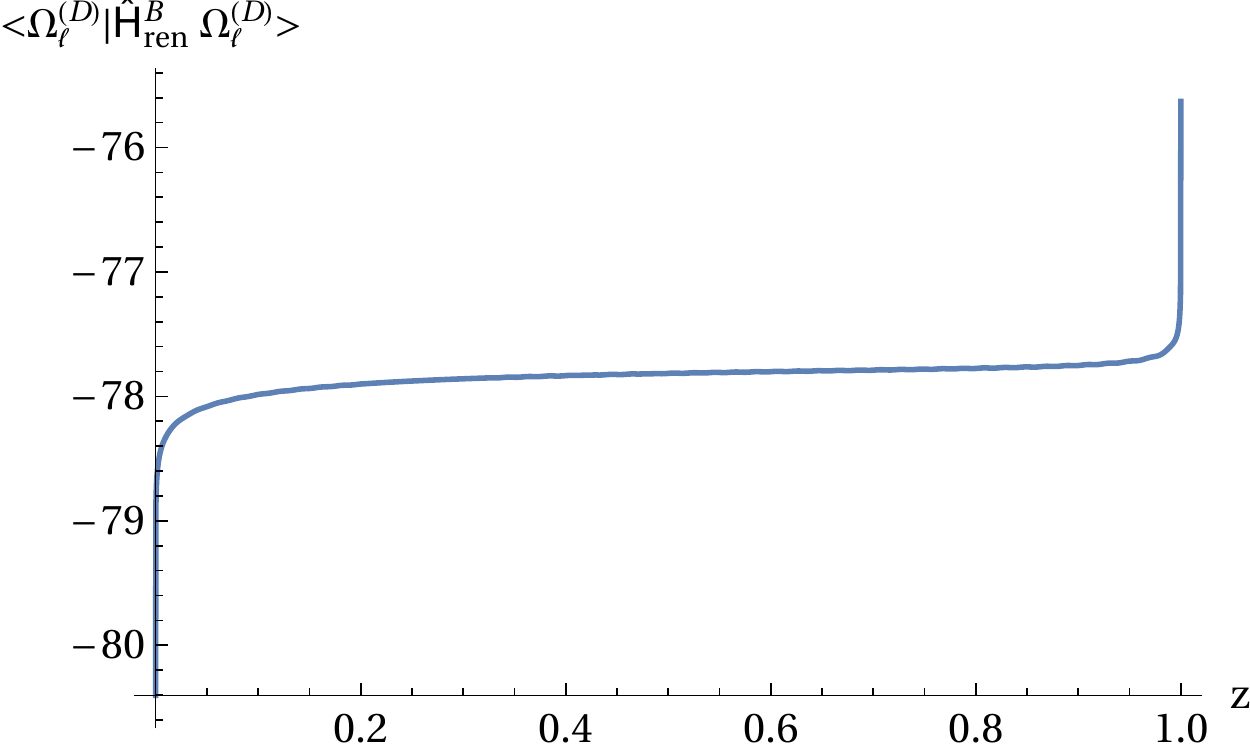}
\caption{Case $\beta_2' = -0.05$. The boundary Casimir energy is $\langle \Omega_\ell^{({\rm D})} | \hat H^{\partial}_{\rm ren} \Omega_\ell^{({\rm D})} \rangle \approx 0.73$. The integrated bulk Casimir energy is $E_0 \approx -77.84$.}
\label{fig:beta2p0.05}
\end{figure}

\begin{figure}[!ht]
\centering
\includegraphics[width=0.45\textwidth]{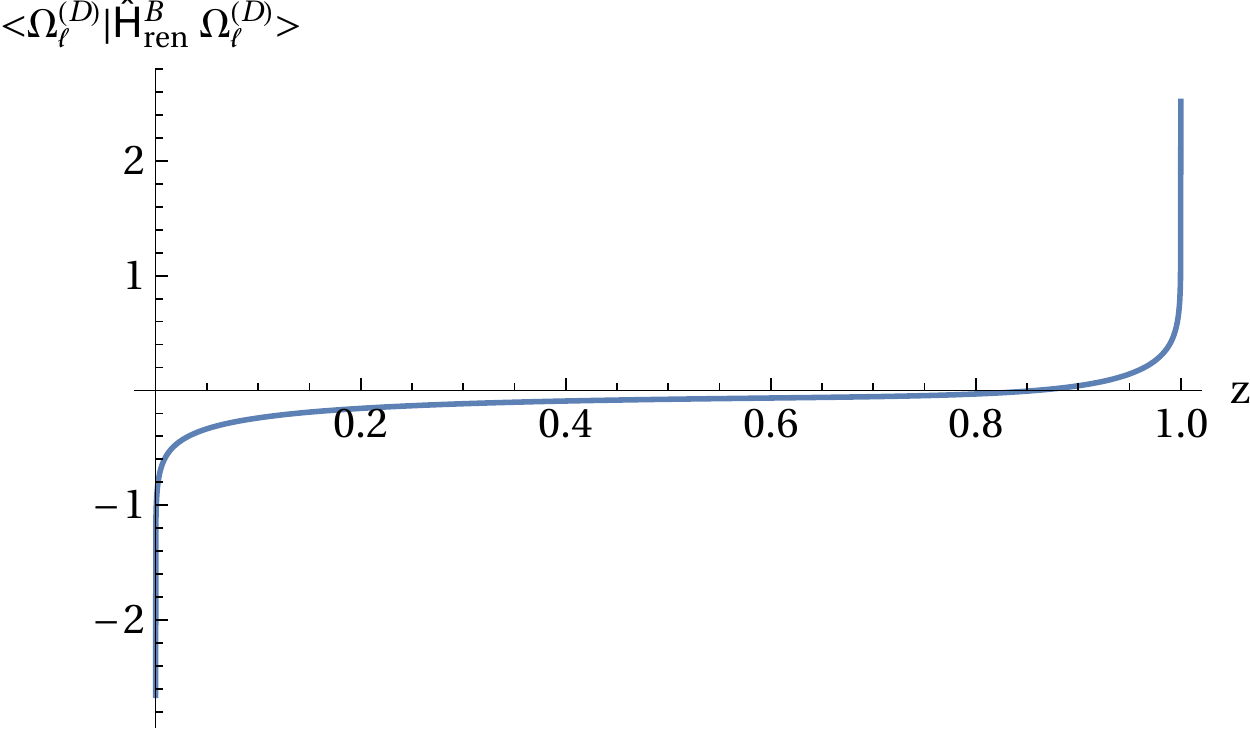}
\caption{Case $\beta_2' = 0.5$. The boundary Casimir energy is $\langle \Omega_\ell^{({\rm D})} | \hat H^{\partial}_{\rm ren} \Omega_\ell^{({\rm D})} \rangle \approx 0.43$. The integrated Casimir energy is $E_0 \approx -0.08$.}
\label{fig:beta2p-0.5}
\end{figure}

\newpage

\section{Final remarks}
\label{Sec:Conclusions}

Throughout this work  we have obtained the renormalized local state polarization and local Casimir energy for a system consisting of a bulk scalar field defined on the interval $[0, \ell]$ coupled  to a boundary observable defined on the right-hand end of the interval. The coupling between the bulk scalar field and the boundary observable  is implemented through a dynamical boundary condition for the bulk field. We have given expressions for the  renormalized bulk  and boundary local state polarization and local Casimir energy at zero temperature in Section \ref{Casimir} and in Section~\ref{CasimirTemp} at positive temperature. Our computations reveal that the renormalized local state   polarization and the local Casimir energy are conserved in time, both for the bulk scalar field and for the boundary observable, and we also show that in the case of the bulk they display logarithmic (integrable) divergences near the boundaries of the interval. 

These divergences occur also for non-dynamical boundary conditions in the Dirichlet and Robin class (see e.g. \cite{Fulling:1989nb}), and we stress that they do not obstruct the total (integrated) Casimir energy from being finite, since they are integrable. Indeed, in Section~\ref{sec:Numerics} we have explored numerically a sample of cases and obtained an approximation of the integrated energy numerically.

To the best of our knowledge, the Casimir energy for fields with dynamical boundary conditions had only been studied in \cite{Fosco:2013wpa} at zero temperature, but in that case only the integrated Casimir energy was obtained, and no information on the local Casimir energy, which we study here, can be directly inferred from their computational method.

The work that we have carried out here can be generalized in several directions. First, by using Fourier transform methods, one can generalize the case of the interval to $n$ dimensions. In this case, one has field theory in the bulk and in the boundary. Second,  one can study other species of linear fields, such as the Maxwell, Dirac or Proca fields. Third, we have only studied the static Casimir effect, but the dynamical Casimir effect \cite{Dodonov:2020eto} should also be of great interest. See, e.g., refs. \cite{Davies:1976ei, Good:2013lca, Juarez-Aubry:2014jba} for the dynamical Casimir effect in the context of quantum field theory and for analogies with black hole radiation and formation. It is particularly interesting the case in which the coefficients $\beta_1$, $\beta_1'$, $\beta_2$ and $\beta_2'$ are time-dependent \cite{Fosco:2013wpa}, \cite{Wilson:2011}, in which case one should have so-called creation of particles.

%
%
%
%

\section*{Acknowledgments}
Benito A. Ju\'arez-Aubry is supported by a DGAPA-UNAM Postdoctoral Fellowship. This paper was partially written while Ricardo Weder was visiting the Institut de Math\'ematique d'Orsay, Universit\'e  Paris-Saclay.  Ricardo Weder thanks Christian G\'erard for his kind hospitality.

\appendix

\section{Useful formulae}
\label{app:UsefulFormulae}

Let $\Lambda, k \in \mathbb{R}$. We have the following formula
\begin{align}
\sum_{n = 1}^\infty \frac{\ee^{\ii \Lambda n}}{n^k} = {\rm Li}_k\left(\ee^{\ii \Lambda}\right),
\end{align}
where ${\rm Li}_k$ is the polylogarithm or Jonqui\`ere function of order $k$, see e.g. \cite[Eq. 25.12.10]{NIST}. In particular, throughout the paper we will make use of the following three formulae for $k = 1$, $k = 0$ and $k = -1$,
\begin{subequations}
\begin{align}
\sum_{n = 1}^\infty \frac{\ee^{\ii \Lambda n}}{n} & = - \ln \left(1 - \ee^{\ii \Lambda} \right), \label{sumJonq1} \\
\sum_{n = 1}^\infty \ee^{\ii \Lambda n} & = \frac{\ee^{\ii \Lambda }}{1-\ee^{\ii \Lambda }}, \label{sumJonq0}\\
\sum_{n = 1}^\infty n \, \ee^{\ii \Lambda n} & = \frac{\ee^{\ii \Lambda }}{\left(1-\ee^{\ii \Lambda }\right)^2}. \label{sumJonq-1}
\end{align}
\end{subequations}

Here, $\ln z$ is the principal branch of the logarithm with the argument of $z$ in $(-\pi, \pi).$ Equation \eqref{sumJonq1} can be easily obtained from the integral representation of ${\rm Li}_1$ \cite[Eq. 25.12.11]{NIST}, whereby
\begin{align}
\sum_{n = 1}^\infty \frac{\ee^{\ii \Lambda n}}{n} = {\rm Li}_1\left(\ee^{\ii \Lambda}\right) = \ee^{\ii \Lambda} \int_0^\infty \! \dd x \frac{1}{\ee^x - \ee^{\ii \Lambda}} = - \ln \left(1 - \ee^{\ii \Lambda} \right),
\label{Jonq1Int}
\end{align}
while  \eqref{sumJonq0} can be obtained from  \eqref{sumJonq1} by taking derivatives in $\Lambda$ on the left-hand and right-hand sides. Similarly,  \eqref{sumJonq-1}   
is obtained taking the derivative in $\Lambda$ on the left-hand and right-hand sides of \eqref{sumJonq0}.
Let $x \in \mathbb{R}$. It holds that
\begin{equation}
\left|\ee^{\ii x} - \sum_{k = 0}^{n} \frac{(\ii x)^k}{k!} \right| \leq \frac{|x|^{n+1}}{(n+1)!},
\label{LemmaBound}
\end{equation}
for any $n \in \mathbb{N}_0$. This follows from estimating the Lagrange Remainder in the Taylor series for $\ee^{\ii x}$ to order $n$, cf. formulae 0.317.2 and 0.317.3 in \cite{Gradshteyn7}.

\section{Asymptotic estimates for the eigenvalues}
\label{App:Estimates}

In this appendix, we obtain asymptotic formulae for the eigenvalues, $\omega^2_n, n=1,\dots,$ of our  classical problem in Section~\ref{sec:Class},  with $V(z) = 0$ in the cases where there is a Dirichlet or Robin boundary condition at $z = 0$. The results that we present here are an improvement on the asymptotics presented in \cite[Sec. 4]{fulton:1977}. 
We assume throughout this appendix that $\rho > 0$ and $\beta_2' \neq 0$. The detailed assumptions of Proposition~\ref{pos} , which ensure that the eigenvalues for problem \eqref{w.2} be all positive, are not needed here. The case for $\beta_2' = 0$ can also be treated, but estimates like the ones that we shall present in this appendix should be obtained separately, and we will not be concerned with this case.  Furthermore, we use  use the dispersion relation, $\omega_n^2 = s_n^2 + m^2, n=1,\dots,$ to write \eqref{w.2} as follows,
\begin{align}
\label{1+1massless}
\left\{
                \begin{array}{l}
                  -\partial^2_z \varphi = s_n^2 \varphi, \\
                  \cos \alpha \varphi(0) + \sin \alpha   \partial_z \varphi(0) = 0, \, \alpha \in [0, \pi), \\
                  -\left[(\beta_1 + \beta_1' m^2) \varphi(\ell) - (\beta_2 + \beta_2' m^2) \partial_z\varphi(\ell) \right] = s_n^2 \left[ \beta_1' \varphi(\ell) - \beta_2' \partial_z\varphi(\ell) \right].
\end{array}
\right.
\end{align}

We are concerned with obtaining large eigenvalue estimates for problem \eqref{1+1massless}. 

In order to obtain the eigenvalues, we seek to find the roots of the Wronskian, which, following standard techniques \cite{Titchmarsh:1946} in the problem at hand, can be written as  (see \cite[Eq. (3.4)]{fulton:1977})
\begin{equation}
\omega(s^2) = (\beta_1' s^2 + \beta_1 + \beta_1' m^2) \phi_{s^2}(\ell) - (\beta_2' s^2
 + \beta_2 + \beta_2' m^2) \partial_z{\phi}_{s^2}(\ell),
 \label{Wronsk}
\end{equation}
where the eigenfunctions $\phi_{s^2}$ satisfy
\begin{align}
\label{BC0}
                  -\partial^2_z \phi_{s^2} & = s^2  \phi_{s^2},  \\
\begin{pmatrix}
           \phi_{s^2}(0)  \\
           \partial_z\phi_{s^2}(0)
         \end{pmatrix} 
& = \begin{pmatrix}
           \sin\alpha  \\
           -\cos\alpha
         \end{pmatrix}, 
\end{align}
with $\alpha = 0$ for Dirichlet boundary condition, $\alpha \in ( 0, \pi)$ for Robin boundary condition, at $z = 0,$ and, in particular, $\alpha = \pi/2$ in the Neumann case. 

In the case of Dirichlet boundary condition, we solve  \eqref{BC0} with $\phi_{s^2}(z) = -\sin (s z)/s$. For Robin boundary condition, we set $\phi_{s^2}(z) = - \cos \alpha \sin (s z)/s +  \sin \alpha \cos( s z)$ (in particular, for the Neumann case $\phi_{s^2}(z) = \cos (s z)$).

We henceforth use labels ${\rm D}$ and ${\rm R}$ for Dirichlet and Robin quantities, respectively.

\subsection{Dirichlet boundary condition at $z = 0$}

Following \cite[Sec. 4]{fulton:1977}, Rouch\'e's theorem (see e.g. \cite{Conway:1978}) implies that, for sufficiently large $n$,
\begin{equation}
s_n^{\rm D} = \frac{(n-1/2) \pi}{\ell} + \delta_n^{\rm D}, \hspace{1cm} \delta_n^{\rm D} = O\left( n^{-1}\right).
\label{snD1}
\end{equation} 

Using  \eqref{Wronsk}, we need to solve for
\begin{equation}
(\beta_1' (s_n^{\rm D})^2 + \beta_1 + \beta_1' m^2) \sin (s_n^{\rm D} \ell) - s_n^{\rm D} (\beta_2' (s_n^{\rm D})^2
 + \beta_2 + \beta_2' m^2) \cos (s_n^{\rm D} \ell) = 0,
\end{equation}
which upon using  \eqref{snD1} can be written as
\begin{align}
& \left[\beta_1' \left(\frac{(n-1/2) \pi}{\ell} + \delta_n^{\rm D}\right)^2 + \beta_1 + \beta_1' m^2\right] \cos ( \delta_n^{\rm D} \ell) \nonumber \\
& + \left( \frac{(n-1/2) \pi}{\ell} + \delta_n^{\rm D} \right) \left[\beta_2' \left(\frac{(n-1/2) \pi}{\ell} + \delta_n^{\rm D}\right)^2  + \beta_2 + \beta_2' m^2\right] \sin (\delta_n^{\rm D} \ell) = 0.
\label{WronskDiri}
\end{align}

Equation \eqref{WronskDiri} can be expanded using the fact that $\delta_n^{\rm D} = O\left( n^{-1}\right)$, and one obtains that
\begin{equation}
\beta_1' + \pi \beta_2' (n-1/2)  \delta_n^{\rm D} = O\left(n^{-2}\right),
\end{equation}
from where it follows that
\begin{equation}
s_n^{\rm D} = \frac{(n-1/2) \pi}{\ell} + \delta_n^{\rm D}, \hspace{1cm} \delta_n^{\rm D} = -\frac{\beta_1'}{\beta_2'  (n-1/2) \pi} + \epsilon_n^{\rm D}, \hspace{1cm} \epsilon_n^{\rm D} = O\left( n^{-3}\right).
\label{snD2}
\end{equation} 

Inserting  \eqref{snD2} into  \eqref{WronskDiri}, we find the relation
\begin{align}
\beta_1 - \frac{(\beta_1')^3}{3 (\beta_2')^2} + \frac{(\beta_1')^2}{\ell \beta_2'} - \frac{\beta_1' \beta_2}{\beta_2'} + \frac{\pi^3 \beta_2'(n-1/2)^3 \epsilon_n^{\rm D}}{\ell^2} = O \left(n^{-2} \right),
\end{align}
which can be solved to yield
\begin{align}
s_n^{\rm D} & = \frac{(n-1/2) \pi}{\ell} + \delta_n^{\rm D}, \nonumber \\
\delta_n^{\rm D} & = -\frac{\beta_1'}{\beta_2' \pi (n-1/2)} + \frac{ \ell^2 \left(-3 \beta_1 (\beta_2')^2 + 3 \beta_2 \beta_1' \beta_2' + (\beta_1')^3\right)- 3 \ell (\beta_1')^2 \beta_2'}{3 \pi ^3 (\beta_2')^3(n-1/2)^3} + O\left( n^{-5}\right).
\label{snD3}
\end{align} 

The expansion \eqref{snD3} suffices for our purposes, but one can recursively obtain more precise estimates.

\subsection{Robin boundary condition at $z = 0$}

The treatment is analogous to the Dirichlet case. Following \cite[Sec. 4]{fulton:1977}, Rouch\'e's theorem (see e.g. \cite{Conway:1978}) implies that, for sufficiently large $n$,
\begin{equation}
s_n^{\rm R} = \frac{(n-1) \pi}{\ell} + \delta_n^{\rm R}, \hspace{1cm} \delta_n^{\rm R} = O\left( n^{-1}\right).
\label{snR1}
\end{equation} 

Using  \eqref{Wronsk}, we need to solve for
\begin{align}
& (\beta_1' (s_n^{\rm R})^2 + \beta_1 + \beta_1' m^2) \left[-\frac{\cos \alpha \sin (s_n^{\rm R} \ell) }{s_n^{\rm R}} + \sin \alpha \cos(s_n^{\rm R} \ell)  \right] \nonumber \\
& + s_n^{\rm R} (\beta_2' (s_n^{\rm R})^2  + \beta_2 + \beta_2' m^2) \left[ \cos \alpha \cos (s_n^{\rm R} \ell) + s_n^{\rm R} \sin \alpha \sin (s_n^{\rm R} \ell) \right] = 0,
\end{align}
which upon using  \eqref{snR1} can be written as
\begin{align}
& \left(\beta_1' \left(\frac{(n-1) \pi}{\ell} + \delta_n^{\rm R}\right)^2 + \beta_1 + \beta_1' m^2\right) \left[-\frac{\cos \alpha \sin (\delta_n^{\rm R} \ell) }{\frac{(n-1) \pi}{\ell} + \delta_n^{\rm R}} + \sin \alpha \cos(\delta_n^{\rm R} \ell)  \right] \nonumber \\
& + \left(\frac{(n-1) \pi}{\ell} + \delta_n^{\rm R}\right) \left(\beta_2' \left(\frac{(n-1) \pi}{\ell} + \delta_n^{\rm R}\right)^2  + \beta_2 + \beta_2' m^2\right) \nonumber \\
& \times \left[ \cos \alpha \cos (\delta_n^{\rm R} \ell) + \left(\frac{(n-1) \pi}{\ell} + \delta_n^{\rm R}\right) \sin \alpha \sin (\delta_n^{\rm R} \ell) \right] = 0.
\label{WronskRobin}
\end{align}

Equation \eqref{WronskRobin} can be expanded using the fact that $\delta_n^{\rm R} = O\left( n^{-1}\right)$, and one obtains that
\begin{equation}
\sin \alpha  (\pi  \delta_n^{\rm R}(n-1) \beta_2'+\beta_1')+\beta_2' \cos \alpha = O(n^{-2}) ,
\end{equation}
from where it follows that
\begin{equation}
s_n^{\rm R} = \frac{(n-1) \pi}{\ell} + \delta_n^{\rm R}, \hspace{1cm} \delta_n^{\rm R} = -\frac{\beta_1' + \beta_2' \cot \alpha}{\beta_2'  (n-1) \pi} + \epsilon_n^{\rm R}, \hspace{1cm} \epsilon_n^{\rm R} = O\left( n^{-3}\right).
\label{snR2}
\end{equation} 

Iterating as in the Dirichlet case by inserting  \eqref{snR2} into  \eqref{WronskRobin}, we find that
\begin{align}
s_n^{\rm R} & = \frac{(n-1) \pi}{\ell} + \delta_n^{\rm R}, \nonumber \\
\delta_n^{\rm R} & = -\frac{\beta_1' + \beta_2' \cot \alpha}{\beta_2'  (n-1) \pi} + \left( \frac{ \ell^2 \left(-3 \beta_1 (\beta_2')^2 + 3 \beta_2 \beta_1' \beta_2' + (\beta_1')^3\right)- 3 \ell (\beta_1')^2 \beta_2'}{3 \pi ^3 (\beta_2')^3} \right. \nonumber \\
& \left. + \frac{\ell \cot \alpha [-6 \beta_1' + \beta_2' \cot \alpha (3-\ell \cot \alpha)]}{3 \pi ^3 \beta_2'} \right) \frac{1}{(n-1)^3} + O\left( n^{-4}\right).
\label{snR3}
\end{align} 

In the case of Neumann boundary conditions, setting $\alpha = \pi/2$ in the Robin case, one has the expansion
\begin{align}
s_n^{\rm N} & = \frac{(n-1) \pi}{\ell} + \delta_n^{\rm N}, \nonumber \\
\delta_n^{\rm N} & = -\frac{\beta_1' }{\beta_2'  (n-1) \pi} + \frac{ \ell^2 \left(-3 \beta_1 (\beta_2')^2 + 3 \beta_2 \beta_1' \beta_2' + (\beta_1')^3\right)- 3 \ell (\beta_1')^2 \beta_2'}{3 \pi ^3 (\beta_2')^3 (n-1)^3}   + O\left( n^{-4}\right).
\label{snN3}
\end{align}

\section{Calculations with a Robin boundary condition at $z = 0$}
\label{app:Rob}

We present the computations for the renormalized local state  polarization and Casimir energy in the case in which a Robin boundary condition is imposed at $z = 0$.

\subsection{Renormalized local state  polarization in the bulk}
\label{app:RobVB}

We are interested in extracting the coincidence-limit singular behavior of
\begin{subequations}
\begin{align}
& \langle  \Omega_\ell^{({\rm R})} | \hat \Phi^{\rm B}(t,z) \hat \Phi^{\rm B}(t',z') \Omega_\ell^{({\rm R})} \rangle = \sum_{n = 1}^\infty \left[ R_{{\rm B}\, n}^{{\rm V}(1)}(t,t',z,z') +  R_{{\rm B}\, n}^{{\rm V}(2)}(t,t',z,z') +  R_{{\rm B}\, n}^{{\rm V}(3)}(t,t',z,z') \right], \label{SingRobPolBulk} \\
& R_{{\rm B}\, n}^{{\rm V}(1)}(t,t',z,z') := \frac{(\mathcal{N}^{\rm R}_n)^2}{8 \omega_n^{\rm R}} \ee^{-\ii \omega_n^{\rm R} (t-t')}    \left(\sin^2 \alpha +  \frac{\cos^2 \alpha}{(s_n^{\rm R})^2} \right) \left(\ee^{\ii s_n^{\rm R}(z-z')} + \ee^{-\ii s_n^{\rm R}(z-z')}\right), \label{RBV1} \\
& R_{{\rm B}\, n}^{{\rm V}(2)}(t,t',z,z') := \frac{(\mathcal{N}^{\rm R}_n)^2}{8 \omega_n^{\rm R}} \ee^{-\ii \omega_n^{\rm R} (t-t')}    \left(\sin^2 \alpha - \frac{\cos^2 \alpha}{(s_n^{\rm R})^2} \right) \left(\ee^{\ii s_n^{\rm R}(z+z')} + \ee^{-\ii s_n^{\rm R}(z+z')}\right), \label{RBV2}\\
& R_{{\rm B}\, n}^{{\rm V}(3)}(t,t',z,z') := - \frac{(\mathcal{N}^{\rm R}_n)^2}{4 \omega_n^{\rm R}} \ee^{-\ii \omega_n^{\rm R} (t-t')}  \frac{\sin \alpha \cos \alpha}{ \ii s_n^{\rm R}}  \left( \ee^{\ii s_n^{\rm R}(z+z')} - \ee^{-\ii s_n^{\rm R}(z+z')} \right). \label{RBV3}
\end{align}
\end{subequations}

It can be seen from the asymptotic estimates \eqref{snR3} that the sum defined by the summand $R_{{\rm B}\, n}^{{\rm V}(1)}$ fails to converge when $t' = t$ and $z' = z$. Let us study this sum in detail. First, we rewrite
\begin{align}
 \sum_{n = 1}^\infty R_{{\rm B}\, n}^{{\rm V}(1)}(t,t',z,z') & = \sum_{n = 1}^\infty \frac{(\mathcal{N}^{\rm R}_n)^2}{8 \omega_n^{\rm R}} \left(\sin^2 \alpha + (s_n^{\rm R})^{-2} \cos^2 \alpha \right) \left\{ \ee^{-\ii \omega_n^{\rm R} (t-t')}  \left(\ee^{\ii s_n^{\rm R}(z-z')} + \ee^{-\ii s_n^{\rm R}(z-z')}\right) \right. \nonumber \\
& -\left.  \left(\ee^{\frac{\ii \pi}{\ell}(n-1) [-(t-t') + (z-z')] } + \ee^{\frac{\ii \pi}{\ell}(n-1) [-(t-t') - (z-z')] }\right)\right\} \nonumber \\
& + \sum_{n = 1}^\infty \frac{(\mathcal{N}^{\rm R}_n)^2}{8 \omega_n^{\rm R}} \left(\sin^2 \alpha + (s_n^{\rm R})^{-2} \cos^2 \alpha \right) \left(\ee^{\frac{\ii \pi}{\ell}(n-1) [-(t-t') + (z-z')] } + \ee^{\frac{\ii \pi}{\ell}(n-1) [-(t-t') - (z-z')] }\right).
\label{SingRobPolBulk1}
\end{align}

The first sum on the right-hand side of  \eqref{SingRobPolBulk1} vanishes in the limit $(t',z') \to (t,z)$, since by estimate \eqref{LemmaBound}
\begin{align}
 \sum_{n = 1}^\infty & \left| \frac{(\mathcal{N}^{\rm R}_n)^2}{8 \omega_n^{\rm R}} \left(\sin^2 \alpha + (s_n^{\rm R})^{-2} \cos^2 \alpha \right) \left(\ee^{\ii [-\omega_n^{\rm R} (t-t') \pm s_n^{\rm R}(z-z')]} - \ee^{\frac{\ii \pi}{\ell}(n-1) [-(t-t') \pm (z-z')] } \right) \right| \nonumber \\
& \leq \sum_{n = 1}^\infty \left| \frac{(\mathcal{N}^{\rm R}_n)^2}{8 \omega_n^{\rm R}} \left(\sin^2 \alpha + (s_n^{\rm R})^{-2} \cos^2 \alpha \right)  \left[-\left(\omega_n^{\rm R} - \frac{\pi}{\ell}(n-1) \right) (t-t') \pm  \left(s_n^{\rm R} - \frac{\pi}{\ell}(n-1) \right)(z-z')\right] \right|. 
\label{DCRob1}
\end{align}
The summand on the right-hand side of \eqref{DCRob1} vanishes as $(t',z') \to (t,z),$ and by \eqref{snR3} it is $O(n^{-2})$ uniformly for $t,t'$ in bounded sets and  $z,z' \in [0,\ell].$ Hence, by dominated convergence, the first sum on the right-hand side of  \eqref{SingRobPolBulk1} vanishes at $ (t',z') \to (t,z).$ The second sum on the right-hand side of  \eqref{SingRobPolBulk1} is therefore the only contribution on the coincidence limit. We rewrite it as
\begin{align}
& \sum_{n = 1}^\infty \frac{(\mathcal{N}^{\rm R}_n)^2}{8 \omega_n^{\rm R}} \left(\sin^2 \alpha + (s_n^{\rm R})^{-2} \cos^2 \alpha \right) \left(\ee^{\frac{\ii \pi}{\ell}(n-1) [-(t-t') + (z-z')] } + \ee^{\frac{\ii \pi}{\ell}(n-1) [-(t-t') - (z-z')] }\right) \nonumber \\
& = \sum_{n = 1}^\infty  \left(\frac{(\mathcal{N}^{\rm R}_n)^2}{8 \omega_n^{\rm R}} \left(\sin^2 \alpha + (s_n^{\rm R})^{-2} \cos^2 \alpha\right) - \frac{1}{4 \pi  n} \right) \left(\ee^{\frac{\ii \pi}{\ell}(n-1) [-(t-t') + (z-z')] } + \ee^{\frac{\ii \pi}{\ell}(n-1) [-(t-t') - (z-z')] }\right) \nonumber \\
& +\sum_{n = 1}^\infty  \frac{1}{4 \pi  n} \left(\ee^{\frac{\ii \pi}{\ell}(n-1) [-(t-t') + (z-z')] } + \ee^{\frac{\ii \pi}{\ell}(n-1) [-(t-t') - (z-z')] }\right).
\label{SingRobPolBulk1-2}
\end{align}

The first sum on the right-hand side of  \eqref{SingRobPolBulk1-2} converges in absolute value and uniformly in $t,t'$ and $z, z'$, since by the asymptotics in  \eqref{snR3}
\begin{align}
\frac{(\mathcal{N}^{\rm R}_n)^2}{8 \omega_n^{\rm R}} \left(\sin^2 \alpha \pm (s_n^{\rm R})^{-2} \cos^2 \alpha \right) = \frac{1}{4 \pi  n} +O\left(n^{-2}\right),
\label{uniftzRBVn1}
\end{align}
so one can take the limit inside the sum, while the second sum in  \eqref{SingRobPolBulk1-2} can be obtained explicitly, using formula \eqref{sumJonq1},
\begin{align}
 \sum_{n = 1}^\infty  \frac{1}{4 \pi  n}  \left(\ee^{\frac{\ii \pi}{\ell}(n-1) [-(t-t') + (z-z')] } + \ee^{\frac{\ii \pi}{\ell}(n-1) [-(t-t') - (z-z')] }\right) & = -\frac{1}{4 \pi } \left[\ee^{\frac{\ii \pi}{\ell}[(t-t')-(z-z')]} \ln \left(1-\ee^{\frac{\ii \pi}{\ell}[-(t-t')+(z-z')]}\right) \right. \nonumber \\
& \left. +\ee^{\frac{\ii \pi}{\ell}[(t-t')+(z-z')]} \ln \left(1-\ee^{-\frac{\ii \pi}{\ell}[(t-t')+(z-z')]}\right)\right],
\label{SingRobPolBulk1-3}
\end{align}
and diverges logarithmically in the coincidence limit, as can be seen from  \eqref{SingRobPolBulk1-3}. Indeed, subtracting the Hadamard bi-distribution  \eqref{HM} we obtain the finite limit
\begin{align}
 \lim_{(t',z') \to (t,z)} \left[\sum_{n = 1}^\infty R_{{\rm B}\, n}^{{\rm V}(1)}(t,t',z,z') - H_{\rm M} ((t,z),(t',z')) \right] & = \frac{1}{4 \pi} \ln \left( \frac{m^2 \ell^2}{4 \pi^2} \right) + \frac{\gamma}{2 \pi} \nonumber \\
& + \sum_{n = 1}^\infty  \left(\frac{(\mathcal{N}^{\rm R}_n)^2}{4 \omega_n^{\rm R}} \left(\sin^2 \alpha + (s_n^{\rm R})^{-2} \cos^2 \alpha\right) - \frac{1}{2 \pi  n} \right).
\label{RBV1lim}
\end{align}

With the aid of estimate \eqref{LemmaBound} and \eqref{snR3}, we can apply the dominated convergence theorem by similar arguments to the ones used for studying the term \eqref{RBV1} to obtain that
\begin{align}
 \lim_{(t',z') \to (t,z)} \sum_{n = 1}^\infty  R_{{\rm B}\, n}^{{\rm V}(2)}(t,t',z,z') & = \sum_{n = 1}^\infty \frac{(\mathcal{N}^{\rm R}_n)^2}{4 \omega_n^{\rm R}}  \left(\sin^2 \alpha - (s_n^{\rm R})^{-2} \cos^2 \alpha \right)  \left[ \cos(2 s_n^{\rm R} z)  - \cos \left(  \frac{2(n-1)\pi}{\ell} z \right) \right] \nonumber \\
& + \lim_{(t',z') \to (t,z)} \sum_{n = 1}^\infty \left[R_{{\rm B}\, n}^{{\rm V}(2)+}(t,t',z,z') + R_{{\rm B}\, n}^{{\rm V}(2)-}(t,t',z,z') \right],
\label{RBV2-2}
\end{align}
with
\begin{align}
 R_{{\rm B}\, n}^{{\rm V}(2)\pm}(t,t',z,z') & :=  \frac{(\mathcal{N}^{\rm R}_n)^2}{8 \omega_n^{\rm R}}  \left(\sin^2 \alpha - (s_n^{\rm R})^{-2} \cos^2 \alpha \right) \ee^{\ii \frac{(n-1) \pi}{\ell} \left[-(t-t')\pm(z+z') \right]}.
\end{align}

Writing the second sum on the right-hand side of  \eqref{RBV2-2} as
\begin{align}
& \lim_{(t',z') \to (t,z)} \sum_{n = 1}^\infty \left[R_{{\rm B}\, n}^{{\rm V}(2)+}(t,t',z,z') + R_{{\rm B}\, n}^{{\rm V}(2)-}(t,t',z,z') - \frac{1}{4 \pi n}\left(  \ee^{\ii \frac{(n-1) \pi}{\ell} \left[-(t-t')+(z+z') \right]} \right. \right. \nonumber \\
& \left. \left. + \ee^{\ii \frac{(n-1) \pi}{\ell} \left[-(t-t')-(z+z') \right]} \right)\right] + \lim_{(t',z') \to (t,z)} \sum_{n = 1}^\infty \frac{1}{4 \pi n}\left(  \ee^{\ii \frac{(n-1) \pi}{\ell} \left[-(t-t')+(z+z') \right]} + \ee^{\ii \frac{(n-1) \pi}{\ell} \left[-(t-t')-(z+z') \right]} \right),
\end{align}
one can use the dominated convergence theorem in view of  \eqref{uniftzRBVn1} to finally obtain that
\begin{align}
\lim_{(t',z') \to (t,z)} \sum_{n = 1}^\infty  R_{{\rm B}\, n}^{{\rm V}(2)}(t,t',z,z')&  = \sum_{n = 1}^\infty \left[ \frac{(\mathcal{N}^{\rm R}_n)^2}{4 \omega_n^{\rm R}} \left(\sin^2 \alpha - (s_n^{\rm R})^{-2} \cos^2 \alpha \right) \cos(2 s_n^{\rm R} z) \right.  \nonumber \\
&   \left.  - \frac{1}{ 2 \pi n} \cos \left( \frac{2(n-1) \pi}{\ell} z \right) \right] - \frac{1}{2 \pi} \Re \left[\ee^{-\ii \frac{ 2 \pi }{\ell} z} \ln \left(1-\ee^{\ii  \frac{2 \pi}{\ell} z}\right) \right],
\label{RBV2-lim}
\end{align}
where we have used  \eqref{sumJonq1} to obtain the last term on the right-hand side of  \eqref{RBV2-lim}.

The third term in  \eqref{SingRobPolBulk}, defined by \eqref{RBV3}, can be seen to converge absolutely (using  \eqref{snR3}) and its limit as $(t',z') \to (t,z)$ can be applied to the summand by dominated convergence, whereby one obtains that
\begin{align}
\lim_{(t',z') \to (t,z)} \sum_{n = 1}^\infty R_{{\rm B}\, n}^{{\rm V}(3)}(t,t',z,z') & = - \sum_{n = 1}^\infty \frac{(\mathcal{N}^{\rm R}_n)^2}{2 \omega_n^{\rm R} s_n^{\rm R}}\sin \alpha \cos \alpha      \sin\left(2 s_n^{\rm R} z \right)
\label{RBV3lim}
\end{align}

Adding up  \eqref{RBV1lim}, \eqref{RBV2-lim} and \eqref{RBV3lim}, we finally obtain that
\begin{align}
& \langle \Omega_\ell^{({\rm R})} |  ( \hat \Phi^{{\rm B}}_{\rm ren})^2(t,z) \Omega_\ell^{({\rm R})} \rangle 
 =  \frac{1}{4 \pi} \ln \left( \frac{m^2 \ell^2}{4 \pi^2} \right) +   \frac{\gamma}{2 \pi} - \frac{1}{2 \pi}\Re \left( \ee^{-\frac{\ii \pi}{\ell} z} \ln \left(1 - \ee^{\frac{\ii 2 \pi}{\ell} z} \right) \right) \nonumber \\
& + \sum_{n = 1}^\infty \left\{ \left[ \frac{(\mathcal{N}^{\rm R}_n)^2}{2 \omega_n^{\rm R}} \left(\sin^2 \alpha \cos^2 \left( s_n^{\rm R} z\right) + \frac{\cos^2 \alpha}{(s_n^{\rm R})^2} \sin^2 \left( s_n^{\rm R} z\right) \right)  - \frac{1}{ \pi n} \cos^2 \left( \frac{\pi}{\ell}(n-1) z\right) \right]   \right. \nonumber \\
& \left. - \frac{(\mathcal{N}^{\rm R}_n)^2}{2 \omega_n^{\rm R} s_n^{\rm R}}\sin \alpha \cos \alpha      \sin\left(2 s_n^{\rm R} z \right)\right\},
\label{App:RobinVacuumBulk}
\end{align}
where the summand on the right-hand side of  \eqref{App:RobinVacuumBulk} is $O(n^{-2})$ and the sum converges absolutely and uniformly in $z \in (0, \ell)$.

\subsection{Renormalized local state polarization in the boundary}
\label{app:RobVb}

We are interested in analyzing the coincidence-limit behavior of
\begin{align}
 \langle \Omega_\ell^{({\rm R})} | \hat \Phi^\partial(t) \hat \Phi^\partial(t') \Omega_\ell^{({\rm R})} \rangle & = \sum_{n = 1}^\infty \frac{(\mathcal{N}_n^{\rm R})^2}{2\omega_n^{\rm R}} \ee^{-\ii \omega_n^{\rm R}(t-t')} \left[\beta_1' \left(\sin \alpha  \cos (\ell s_n^{\rm R})-\frac{\cos \alpha  \sin (\ell s_n^{\rm R})}{s_n^{\rm R}}\right) \right. \nonumber \\
 & \left. +\beta_2' (\cos \alpha \cos (\ell s_n^{\rm R})+s_n^{\rm R} \sin \alpha  \sin (\ell s_n^{\rm R})) \right]^2.
 \label{RobVb0}
\end{align}

From our estimates \eqref{snR3}, we observe that
\begin{align}
& \frac{(\mathcal{N}_n^{\rm R})^2}{2\omega_n^{\rm R}} \left[\beta_1' \left(\sin \alpha  \cos (\ell s_n^{\rm R})-\frac{\cos \alpha  \sin (\ell s_n^{\rm R})}{s_n^{\rm R}}\right) +\beta_2' (\cos \alpha \cos (\ell s_n^{\rm R})+s_n^{\rm R} \sin \alpha  \sin (\ell s_n^{\rm R})) \right]^2 \nonumber \\
 & = \frac{\ell^4 \sin^2(\alpha ) (\beta_1' \beta_2-\beta_1 \beta_2')^2}{ \pi ^5 (\beta_2')^2 n^5}+O\left(n^{-6}\right),
 \label{RobVb1}
\end{align}
so that the expression on the left-hand side of  \eqref{RobVb1} provides a summable, $t, t'$-independent bound for the summand on the right-hand side of  \eqref{RobVb0}, which allows us to use the dominated convergence theorem to write
\begin{align}
& \langle \Omega_\ell^{({\rm R})} | ( \hat  \Phi^{\partial} )^2(t) \Omega_\ell^{({\rm R})} \rangle = \lim_{t' \to t} \langle \Omega_\ell^{({\rm R})} | \hat \Phi^\partial(t) \hat \Phi^\partial(t') \Omega_\ell^{({\rm R})} \rangle \nonumber \\
& = \sum_{n = 1}^\infty \frac{(\mathcal{N}_n^{\rm R})^2}{2\omega_n^{\rm R}} \left[\beta_1' \left(\sin \alpha  \cos (\ell s_n^{\rm R})-\frac{\cos \alpha  \sin (\ell s_n^{\rm R})}{s_n^{\rm R}}\right)  +\beta_2' (\cos \alpha \cos (\ell s_n^{\rm R})+s_n^{\rm R} \sin \alpha  \sin (\ell s_n^{\rm R})) \right]^2,
\label{App:RobinVacuumBoundary}
\end{align}
where the summand on the right-hand side of  \eqref{App:RobinVacuumBulk} is $O(n^{-5})$ and the sum converges absolutely.

\subsection{Local Casimir energy in the bulk}
\label{app:RobHB}

We write the local Casimir energy as
\begin{align}
\langle  \Omega_\ell^{({\rm R})} |  \hat H^{\rm B}(t,z) \Omega_\ell^{({\rm R})} \rangle & = \frac{1}{2} \lim_{(t',z') \to (t,z) }\left[ \left( \partial_t \partial_{t'} + \partial_z \partial_{z'} \right) \langle  \Omega_\ell^{({\rm R})} | \hat \Phi^{\rm B}(t,z) \hat \Phi^{\rm B}(t',z') \Omega_\ell^{({\rm R})} \rangle \right. \nonumber \\ 
& \left. - \left( \partial_t \partial_{t'} + \partial_z \partial_{z'} \right) H_{\rm M} 	((t,z),(t',z')) \right] + \frac{m^2}{2} \langle \Omega_\ell^{({\rm R})} | (\hat{\Phi}^{{\rm B} })^2(t,z) \Omega_\ell^{({\rm R})} \rangle.
\end{align}

We are interested in extracting the coincidence-limit singular behavior of
\begin{subequations}
\begin{align}
& \frac{1}{2} \left( \partial_t \partial_{t'} + \partial_z \partial_{z'} \right) \langle  \Omega_\ell^{({\rm R})} | \hat \Phi^{\rm B}(t,z) \hat \Phi^{\rm B}(t',z') \Omega_\ell^{({\rm R})} \rangle = \sum_{n = 1}^\infty \left[ R_{{\rm B}\, n}^{{\rm H}(1)}(t,t',z,z') \right. \nonumber \\
& \hspace{50pt} \left. +  R_{{\rm B}\, n}^{{\rm H}(2)}(t,t',z,z') +  R_{{\rm B}\, n}^{{\rm H}(3)}(t,t',z,z') \right], \label{SingRobHBulk} \\
& R_{{\rm B}\, n}^{{\rm H}(1)}(t,t',z,z') := \frac{(\mathcal{N}^{\rm R}_n)^2}{16 \omega_n^{\rm R}} \left((\omega_n^{\rm R})^2 + (s_n^{\rm R})^2 \right)     \left(\sin^2 \alpha + (s_n^{\rm R})^{-2} \cos^2 \alpha \right) \nonumber \\
& \hspace{50pt} \times \ee^{-\ii \omega_n^{\rm R} (t-t')} \left(\ee^{\ii s_n^{\rm R}(z-z')} + \ee^{-\ii s_n^{\rm R}(z-z')}\right), \label{RBH1} \\
& R_{{\rm B}\, n}^{{\rm H}(2)}(t,t',z,z') := \frac{(\mathcal{N}^{\rm R}_n)^2}{16 \omega_n^{\rm R}} \left( (\omega_n^{\rm R})^2 - (s_n^{\rm R})^2 \right) \left(\sin^2 \alpha - (s_n^{\rm R})^{-2} \cos^2 \alpha \right) \nonumber \\
& \hspace{50pt} \times \ee^{-\ii \omega_n^{\rm R} (t-t')} \left(\ee^{\ii s_n^{\rm R}(z+z')} + \ee^{-\ii s_n^{\rm R}(z+z')}\right) = \frac{m^2}{2} R_{{\rm B}\, n}^{{\rm V}(2)}(t,t',z,z') , \label{RBH2}\\
& R_{{\rm B}\, n}^{{\rm H}(3)}(t,t',z,z') := - \frac{(\mathcal{N}^{\rm R}_n)^2}{8 \omega_n^{\rm R}} \left((\omega_n^{\rm R})^2 - (s_n^{\rm R})^2 \right)  \frac{\sin \alpha \cos \alpha}{ \ii s_n^{\rm R}}  \ee^{-\ii \omega_n^{\rm R} (t-t')} \left( \ee^{\ii s_n^{\rm R}(z+z')} - \ee^{-\ii s_n^{\rm R}(z+z')} \right) \nonumber \\
& \hspace{50pt} = \frac{m^2}{2} R_{{\rm B}\, n}^{{\rm V}(3)}(t,t',z,z').  \label{RBH3}
\end{align}
\end{subequations}

We start by handling the first term. We note that by estimate \eqref{LemmaBound} 

\begin{align}
& \left| \frac{\left((\omega_n^{\rm R})^2 + (s_n^{\rm R})^2 \right)}{\omega_n^{\rm R}} \left\{\ee^{-\ii \omega_n^{\rm R} (t-t')}\ee^{\pm \ii s_n^{\rm R}(z-z')} -  \ee^{\ii \frac{(n-1) \pi}{\ell} [-(t-t')\pm (z-z')]} \right. \right. \nonumber \\
& \left. \left. \times \sum_{k = 0}^2 \frac{1}{k!}\left[ -\ii \left( \omega_n^{\rm R} - \frac{(n-1) \pi}{\ell} \right) (t-t') \pm \ii \left( s_n^{\rm R} - \frac{(n-1) \pi}{\ell} \right)(z-z') \right]^k \right\}\right| \nonumber \\
& \leq \left| \frac{\left((\omega_n^{\rm R})^2 + (s_n^{\rm R})^2 \right)}{\omega_n^{\rm R}} \left\{  \ee^{-\ii \left( \omega_n^{\rm R} - \frac{(n-1) \pi}{\ell} \right) (t-t') \pm \ii \left( s_n^{\rm R} - \frac{(n-1) \pi}{\ell} \right)(z-z')} \right. \right. \nonumber \\
& - \left. \left. \sum_{k = 0}^2 \frac{1}{k!}\left[ -\ii \left( \omega_n^{\rm R} - \frac{(n-1) \pi}{\ell} \right) (t-t')   \pm \ii \left( s_n^{\rm R} - \frac{(n-1) \pi}{\ell} \right)(z-z') \right]^k \right\}\right| \nonumber \\
& \leq \frac{\left((\omega_n^{\rm R})^2 + (s_n^{\rm R})^2 \right)}{6 \omega_n^{\rm R}}  \left|  -\ii \left( \omega_n^{\rm R} - \frac{(n-1) \pi}{\ell} \right) (t-t')   \pm \ii \left( s_n^{\rm R} - \frac{(n-1) \pi}{\ell} \right)(z-z') \right|^3 \nonumber \\
& \leq \frac{\left((\omega_n^{\rm R})^2 + (s_n^{\rm R})^2 \right)}{6 \omega_n^{\rm R}} \left[  \left| \left( \omega_n^{\rm R} - \frac{(n-1) \pi}{\ell} \right) (t-t') \right|  + \left| \left( s_n^{\rm R} - \frac{(n-1) \pi}{\ell} \right)(z-z') \right| \right]^3.
\label{CasEnBulkBound}
\end{align}

From the above bounds, and using the expansion \eqref{snR3}, a dominated convergence argument (whereby one can find a summable summand -- behaving as $O(n^{-2})$ for large $n,$ and that is uniform for
 $t,t'$ in bounded set and in  $z,z'$ -- that bounds the right-hand side of  \eqref{CasEnBulkBound}) yields that
\begin{align}
& \lim_{(t',z') \to (t,z)} \sum_{n = 1}^\infty \left\{ R_{{\rm B}\, n}^{{\rm H}(1)}(t,t',z,z') - \frac{(\mathcal{N}^{\rm R}_n)^2}{16 \omega_n^{\rm R}} \left((\omega_n^{\rm R})^2 + (s_n^{\rm R})^2 \right)     \left(\sin^2 \alpha + \frac{\cos^2 \alpha}{(s_n^{\rm R})^2}  \right) \ee^{\ii \frac{(n-1) \pi}{\ell} [-(t-t')+ (z-z')]} \right. \nonumber \\
& \left. \times \sum_{k = 0}^2 \frac{1}{k!}\left[ -\ii \left( \omega_n^{\rm R} - \frac{(n-1) \pi}{\ell} \right) (t-t') - \ii \left( s_n^{\rm R} - \frac{(n-1) \pi}{\ell} \right)(z-z') \right]^k \right\} = 0.
\end{align}

We are thus interested in the coincidence limits of sums of the form
\begin{align}
& \sum_{n = 1}^\infty \ee^{\ii \frac{(n-1) \pi}{\ell} [-(t-t')\pm (z-z')]} \frac{(\mathcal{N}^{\rm R}_n)^2}{16 \omega_n^{\rm R}} \left((\omega_n^{\rm R})^2 + (s_n^{\rm R})^2 \right)     \left(\sin^2 \alpha + \frac{\cos^2 \alpha}{(s_n^{\rm R})^2}  \right) \nonumber \\
&  \times   \sum_{k = 0}^2 \frac{1}{k!}\left[ -\ii \left( \omega_n^{\rm R} - \frac{(n-1) \pi}{\ell} \right) (t-t') \pm \ii \left( s_n^{\rm R} - \frac{(n-1) \pi}{\ell} \right)(z-z') \right]^k.
\end{align}

Observing from  \eqref{snR3} that
\begin{subequations}
\begin{align}
& \frac{(\mathcal{N}^{\rm R}_n)^2}{16 \omega_n^{\rm R}} \left((\omega_n^{\rm R})^2 + (s_n^{\rm R})^2 \right)     \left(\sin^2 \alpha + \frac{\cos^2 \alpha}{(s_n^{\rm R})^2}  \right)  \sum_{k = 0}^2 \frac{1}{k!}\left[ -\ii \left( \omega_n^{\rm R} - \frac{(n-1) \pi}{\ell} \right) (t-t') \right. \nonumber \\ 
& \left. \pm \ii \left( s_n^{\rm R} - \frac{(n-1) \pi}{\ell} \right)(z-z') \right]^k - R_{{\rm B} n}^{{\rm H} (1)\, \pm}(t,t',z,z') = O \left( n^{-2} \right),  \\
& R_{{\rm B} n}^{{\rm H} (1) \, \pm}(t,t',z,z') :=  \frac{\pi  n}{4 \ell^2}  +\frac{\ii \left(2 \ii \pi  \beta_2'- \ell \beta_2' m^2 (t-t')  -2 \left(\beta_1' +  \beta_2' \cot \alpha \right) (-(t-t')\pm (z-z')) \right)}{8 \beta_2' \ell^2} \nonumber \\
&  -\frac{\left(\beta_2' \ell m^2 (t-t')+2 \left( \beta_1 + \beta_2' \cot \alpha  \right) (-(t-t') \pm (z-z'))  \right)^2}{32  \pi  (\beta_2')^2 \ell^2 n}  ,
\end{align}
\end{subequations}
a dominated convergence argument allows us to take the coincidence limit inside the sum and
\begin{align}
& \lim_{(t',z') \to (t,z)} \sum_{n = 1}^\infty \left\{ R_{{\rm B}\, n}^{{\rm H}(1)}(t,t',z,z') - \ee^{\ii \frac{(n-1) \pi}{\ell} [-(t-t')+ (z-z')]} R_{{\rm B} n}^{{\rm H} (1) \, +}(t,t',z,z') - \ee^{\ii \frac{(n-1) \pi}{\ell} [-(t-t')- (z-z')]} R_{{\rm B} n}^{{\rm H} (1) \, -}(t,t',z,z') \right\} \nonumber \\
& = \sum_{n = 1}^\infty \left\{  \frac{(\mathcal{N}^{\rm R}_n)^2}{8 \omega_n^{\rm R}} \left((\omega_n^{\rm R})^2 + (s_n^{\rm R})^2 \right)     \left(\sin^2 \alpha + \frac{\cos^2 \alpha}{(s_n^{\rm R})^2}  \right) - \frac{\pi (n-1)}{2 \ell^2}  \right\}.
\label{RBH1-1}
\end{align}
The second and third sum on the left-hand side of  \eqref{RBH1-1} can be obtained in closed form and together have the singular structure of $\frac{1}{2}(\partial_t \partial_{t'} + \partial_z \partial_{z'}) H_{\rm M}((t,z),(t',z'))$, cf  \eqref{HM}, plus $O(1)$ terms in the coincidence limit, from where it follows that

\begin{align}
 \lim_{(t',z') \to (t,z)} & \left\{ \sum_{n = 1}^\infty R_{{\rm B}\, n}^{{\rm H}(1)}(t,t',z,z') - \frac{1}{2}(\partial_t \partial_{t'} + \partial_z \partial_{z'}) H_{\rm M}((t,z),(t',z')) \right\} \nonumber \\
& = \sum_{n = 1}^\infty \left\{  \frac{(\mathcal{N}^{\rm R}_n)^2}{8 \omega_n^{\rm R}} \left((\omega_n^{\rm R})^2 + (s_n^{\rm R})^2 \right)     \left(\sin^2 \alpha + \frac{\cos^2 \alpha}{(s_n^{\rm R})^2}  \right) - \frac{\pi (n-1)}{2 \ell^2}  \right\}  -\frac{\pi }{24 \ell^2}-\frac{\cot \alpha }{2 \pi  \ell}-\frac{\beta_1'}{2 \pi  \beta_2' \ell}+\frac{m^2}{8 \pi }.
\label{RBH1-lim}
\end{align}

Using  \eqref{RBV2-lim} and \eqref{RBV3lim} we obtain respectively
\begin{align}
\lim_{(t',z') \to (t,z)} \sum_{n = 1}^\infty  R_{{\rm B}\, n}^{{\rm H}(2)}(t,t',z,z')&  = \frac{m^2}{2}\sum_{n = 1}^\infty \left[ \frac{(\mathcal{N}^{\rm R}_n)^2}{4 \omega_n^{\rm R}} \left(\sin^2 \alpha - (s_n^{\rm R})^{-2} \cos^2 \alpha \right) \cos(2 s_n^{\rm R} z) \right.  \nonumber \\
&   \left.  - \frac{1}{ 2 \pi n} \cos \left( \frac{2(n-1) \pi}{\ell} z \right) \right] - \frac{m^2}{4 \pi} \Re \left[\ee^{-\ii \frac{ 2 \pi }{\ell} z} \ln \left(1-\ee^{\ii  \frac{2 \pi}{\ell} z}\right) \right].
\label{RBH2-lim}
\end{align}
and
\begin{align}
\lim_{(t',z') \to (t,z)} \sum_{n = 1}^\infty R_{{\rm B}\, n}^{{\rm H}(3)}(t,t',z,z') & = - \frac{m^2}{2} \sum_{n = 1}^\infty \frac{(\mathcal{N}^{\rm R}_n)^2}{2 \omega_n^{\rm R} s_n^{\rm R}}\sin \alpha \cos \alpha      \sin\left(2 s_n^{\rm R} z \right).
\label{RBH3-lim}
\end{align}

Collecting  \eqref{RBH1-lim}, \eqref{RBH2-lim} and \eqref{RBH3-lim}, together with the expression for the renormalized local state polarization, \eqref{App:RobinVacuumBulk} we finally have that

\begin{align}
& \langle  \Omega_\ell^{({\rm R})} |  \hat H^{\rm B}(t,z) \Omega_\ell^{({\rm R})} \rangle = -\frac{\pi }{24 \ell^2}-\frac{\cot \alpha }{2 \pi  \ell}-\frac{\beta_1'}{2 \pi  \beta_2' \ell}+\frac{m^2}{8 \pi }\left[ 1  + \ln \left( \frac{m^2 \ell^2}{4 \pi^2} \right) \right] + \frac{\gamma m^2 }{4 \pi} \nonumber \\
& - \frac{m^2}{2 \pi} \Re \left[\ee^{-\ii \frac{ 2 \pi }{\ell} z} \ln \left(1-\ee^{\ii  \frac{2 \pi}{\ell} z}\right) \right] - m^2 \sum_{n = 1}^\infty \frac{(\mathcal{N}^{\rm R}_n)^2}{2 \omega_n^{\rm R} s_n^{\rm R}}\sin \alpha \cos \alpha      \sin\left(2 s_n^{\rm R} z \right) \nonumber \\
&+\sum_{n = 1}^\infty \left\{  \frac{(\mathcal{N}^{\rm R}_n)^2}{8 \omega_n^{\rm R}} \left((\omega_n^{\rm R})^2 + (s_n^{\rm R})^2 \right)     \left(\sin^2 \alpha + \frac{\cos^2 \alpha}{(s_n^{\rm R})^2}  \right) - \frac{\pi (n-1)}{2 \ell^2}  \right\} \nonumber \\
& + \frac{m^2}{2}\sum_{n = 1}^\infty \left[ \frac{(\mathcal{N}^{\rm R}_n)^2}{4 \omega_n^{\rm R}} \left(\sin^2 \alpha - (s_n^{\rm R})^{-2} \cos^2 \alpha \right) \cos(2 s_n^{\rm R} z)   - \frac{1}{ 2 \pi n} \cos \left( \frac{2(n-1) \pi}{\ell} z \right) \right]  \nonumber \\
& + \frac{m^2}{2} \sum_{n = 1}^\infty  \left[ \frac{(\mathcal{N}^{\rm R}_n)^2}{2 \omega_n^{\rm R}} \left(\sin^2 \alpha \cos^2 \left( s_n^{\rm R} z\right) + \frac{\cos^2 \alpha}{(s_n^{\rm R})^2} \sin^2 \left( s_n^{\rm R} z\right) \right) - \frac{1}{ \pi n} \cos^2 \left( \frac{\pi}{\ell}(n-1) z\right) \right] .
\label{App:RobinHBulk}
\end{align}

The summand on the right-hand side of  \eqref{App:RobinHBulk} is $O(n^{-2})$ and the sum converges absolutely and uniformly in $z \in (0, \ell)$.

\subsection{Local Casimir energy in the boundary}
\label{app:RobHb}

We have defined the Casimir energy in the boundary as
\begin{align}
\langle  \Omega_\ell^{({\rm R})} |  \hat H^{\partial}(t) \Omega_\ell^{({\rm R})} \rangle = \frac{1}{2} \lim_{t' \to t } \left( \beta_1' \partial_t \partial_{t'} -\beta_1 \right) \langle  \Omega_\ell^{({\rm R})} | \hat \Phi^{\partial}(t) \hat \Phi^{\partial}(t') \Omega_\ell^{({\rm R})} \rangle,
\end{align}
so we wish to analyze the coincidence-limit behavior of
\begin{align}
 \frac{1}{2}\left(\beta_1' \partial_t \partial_{t'} -\beta_1 \right) & \langle \Omega_\ell^{({\rm R})} | \hat \Phi^\partial(t) \hat \Phi^\partial(t') \Omega_\ell^{({\rm R})} \rangle \nonumber \\
 & = \sum_{n = 1}^\infty \frac{\left[\beta_1' (\omega_n^{\rm R})^2 -\beta_1\right] (\mathcal{N}_n^{\rm R})^2}{4\omega_n^{\rm R}} \ee^{-\ii \omega_n^{\rm R}(t-t')} \left[\beta_1' \left(\sin \alpha  \cos (\ell s_n^{\rm R})-\frac{\cos \alpha  \sin (\ell s_n^{\rm R})}{s_n^{\rm R}}\right) \right. \nonumber \\
 & \left. +\beta_2' (\cos \alpha \cos (\ell s_n^{\rm R})+s_n^{\rm R} \sin \alpha  \sin (\ell s_n^{\rm R})) \right]^2.
 \label{RobHb0}
\end{align}

We observe that from the asymptotics \eqref{snR3}
\begin{align}
& \frac{\left[\beta_1' (\omega_n^{\rm R})^2 -\beta_1 \right](\mathcal{N}_n^{\rm R})^2}{4\omega_n^{\rm R}} \left[\beta_1' \left(\sin \alpha  \cos (\ell s_n^{\rm R})-\frac{\cos \alpha  \sin (\ell s_n^{\rm R})}{s_n^{\rm R}}\right) +\beta_2' (\cos \alpha \cos (\ell s_n^{\rm R})  +s_n^{\rm R} \sin \alpha  \sin (\ell s_n^{\rm R})) \right]^2  \nonumber \\
& = \frac{\beta_1' \ell^2 \sin^2(\alpha ) (\beta_1' \beta_2-\beta_1 \beta_2')^2}{2 \pi^3 (\beta_2')^2 n^3}+O\left(n^{-4}\right).
 \label{RobHb1}
\end{align}
Hence, \eqref{RobHb1} provides a summable, $t, t'$-independent bound for the summand on the right-hand side of  \eqref{RobHb0}, which allows us to use the dominated convergence theorem to take the $t' \to t$ limit inside the sum defining the boundary Casimir energy as follows
\begin{align}
& \langle \Omega_\ell^{({\rm R})} | \hat H^{\partial}(t) \Omega_\ell^{({\rm R})} \rangle = \lim_{t' \to t} \frac{1}{2}\left(\beta_1' \partial_t \partial_{t'} -\beta_1 \right)\langle \Omega_\ell^{({\rm R})} | \hat \Phi^\partial(t) \hat \Phi^\partial(t') \Omega_\ell^{({\rm R})} \rangle \nonumber \\
& = \sum_{n = 1}^\infty \frac{\left[ \beta_1'(\omega_n^{\rm R})^2 -\beta_1 \right] (\mathcal{N}_n^{\rm R})^2 }{4\omega_n^{\rm R}} \left[\beta_1' \left(\sin \alpha  \cos (\ell s_n^{\rm R})-\frac{\cos \alpha  \sin (\ell s_n^{\rm R})}{s_n^{\rm R}}\right)   +\beta_2' (\cos \alpha \cos (\ell s_n^{\rm R})+s_n^{\rm R} \sin \alpha  \sin (\ell s_n^{\rm R})) \right]^2.
\label{App:RobinHBoundary}
\end{align}

The summand on the right-hand side of  \eqref{App:RobinHBoundary} is $O(n^{-3})$ and the sum converges absolutely.



\end{document}